\documentclass[sigconf, authorversion]{acmart}

\copyrightyear{2026}
\acmYear{2026}
\acmConference[WWW '26]{Proceedings of the ACM Web Conference 2026}{April 13--17, 2026}{Dubai, United Arab Emirates}
\acmBooktitle{Proceedings of the ACM Web Conference 2026 (WWW '26), April 13--17, 2026, Dubai, United Arab Emirates}
\acmDOI{10.1145/3774904.3792196}
\acmISBN{979-8-4007-2307-0/2026/04}

\usepackage[utf8]{inputenc} 
\usepackage{url}            
\usepackage{amsfonts}       
\usepackage{nicefrac}       
\usepackage{microtype}      
\usepackage{amsmath}
\usepackage{algorithm}
\usepackage{algorithmic}
\usepackage{marvosym}
\usepackage{array}
\usepackage{multirow}
\usepackage{enumitem}
\usepackage{makecell}
\usepackage{bm}
\usepackage{nicefrac}
\usepackage{amsthm}
\usepackage{setspace}
\usepackage{subfig}
\setlist{leftmargin=*}
\usepackage{color}
\usepackage{xcolor}
\usepackage{xspace}
\def\red#1{\textcolor{red}{#1}}

\long\def\comment#1{}

\def\ie{$i.e.$}
\def\eg{$e.g.$}

\newcommand{\partitle}[1]{\smallskip \noindent \textbf{#1.}}

\newcommand{\firstpartitle}[1]{\noindent \textbf{#1.}}
\newcommand{\bench}{{\sc LeaFBench}\xspace}
\newcommand{\name}{{\sc ZeroPrint}\xspace}

\newcolumntype{L}[1]{>{\raggedright\arraybackslash}p{#1}}

\newtheorem{definition}{Definition}
\newtheorem{theorem}{Theorem}

\newtheorem{lemma}{Lemma}

\usepackage{tcolorbox}

\newtcolorbox{takeawaybox}{
  colback=gray!20,
  colframe=gray!20,
  coltitle=black,
  arc=4pt,
  boxrule=0.5pt,
  boxsep=2pt,
  left=2pt,
  right=2pt,
  top=2pt,
  bottom=2pt,
  before skip=0.5\baselineskip,
  after skip=0.5\baselineskip
}

\usepackage{tikz}
\newcommand{\chatbox}[2][Templates]{%
\begin{center}
    \begin{tikzpicture}[
            chatbox_inner/.style={
                rectangle, 
                rounded corners, 
                opacity=0, 
                text opacity=1, 
                font=\sffamily\scriptsize,
                text width=0.46\textwidth, 
                text height=9pt, 
                inner xsep=6pt, 
                inner ysep=6pt
            },
           chatbox_prompt_inner/.style={chatbox_inner, align=flush left, xshift=0pt, text height=11pt},
           chatbox_user_inner/.style={chatbox_inner, align=flush left, xshift=0pt},
           chatbox_gpt_inner/.style={chatbox_inner, align=flush left, xshift=0pt},
           chatbox/.style={chatbox_inner, draw=black!25, fill=gray!7, opacity=1, text opacity=0},
           chatbox_prompt/.style={chatbox, align=flush left, fill=gray!1.5, draw=black!30, text height=10pt},
           chatbox_user/.style={chatbox, align=flush left},
           chatbox_gpt/.style={chatbox, align=flush left},
           chatbox2/.style={chatbox_gpt, fill=green!25},
           chatbox3/.style={chatbox_gpt, fill=red!20, draw=black!20},
           chatbox4/.style={chatbox_gpt, fill=yellow!30},
           labelbox/.style={
           rectangle, 
           rounded corners, 
           draw=black!50, 
           font=\sffamily\scriptsize\bfseries, 
           fill=gray!5, 
           inner sep=3pt
           },
        ]

        \node[chatbox_user] (q1)[align=justify, text width=0.45\textwidth] {#2};
        \node[chatbox_user_inner] (q1_text)[align=justify, text width=0.45\textwidth] at (q1) {#2};
        \node[labelbox, anchor=north west, yshift=5pt, xshift=5pt] at (q1.north west) {\textbf{#1}};
    \end{tikzpicture}
\end{center}
}

\definecolor{level1}{RGB}{255, 249, 196}  
\definecolor{level2}{RGB}{255, 241, 118}  
\definecolor{level3}{RGB}{220, 231, 117}  
\definecolor{level4}{RGB}{174, 213, 129}  
\definecolor{level5}{RGB}{139, 195, 74}   

\newcommand{\colorauc}[1]{%
  \ifdim #1pt < 0.6pt \cellcolor{level1!80}#1\else
  \ifdim #1pt < 0.7pt \cellcolor{level2!70}#1\else
  \ifdim #1pt < 0.8pt \cellcolor{level3!70}#1\else
  \ifdim #1pt < 0.9pt \cellcolor{level4!70}#1\else
  \cellcolor{level5!70}#1\fi\fi\fi\fi
}

\newcommand{\colormd}[1]{%
  \ifdim #1pt < 0.5pt \cellcolor{level1!80}#1\else
  \ifdim #1pt < 1.0pt \cellcolor{level2!80}#1\else
  \ifdim #1pt < 2.0pt \cellcolor{level3!80}#1\else
  \ifdim #1pt < 2.5pt \cellcolor{level4!80}#1\else
  \cellcolor{level5!80}#1\fi\fi\fi\fi
}

\begin{document}

\title{Reading Between the Lines: Towards Reliable Black-box LLM Fingerprinting via Zeroth-order Gradient Estimation}


\author{Shuo Shao}
\affiliation{%
  \institution{Zhejiang University}
  \city{Hangzhou}
  \country{China}
}
\email{shaoshuo\_ss@zju.edu.cn}

\author{Yiming Li}
\authornote{Corresponding Author.}
\affiliation{%
  \institution{Nanyang Technological University}
  \city{Singapore}
  \country{Singapore}
}
\email{liyiming.tech@gmail.com}

\author{Hongwei Yao}
\affiliation{%
  \institution{City University of Hong Kong}
  \city{Hong Kong}
  \country{China}
}
\email{yao.hongwei@cityu.edu.hk}

\author{Yifei Chen}
\affiliation{%
  \institution{Zhejiang University}
  \city{Hangzhou}
  \country{China}
}
\email{yifei.chen@zju.edu.cn}

\author{Yuchen Yang}
\affiliation{%
  \institution{Zhejiang University}
  \city{Hangzhou}
  \country{China}
}
\email{ychyang@zju.edu.cn}

\author{Zhan Qin}
\affiliation{%
  \institution{Zhejiang University}
  \city{Hangzhou}
  \country{China}
}
\email{qinzhan@zju.edu.cn}
\renewcommand{\shortauthors}{Shuo Shao et al.}
%

\begin{abstract}
The substantial investment required to develop Large Language Models (LLMs) makes them valuable intellectual property, raising significant concerns about copyright protection. LLM fingerprinting has emerged as a key technique to address this, which aims to verify a model's origin by extracting an intrinsic, unique signature (a ``fingerprint'') and comparing it to that of a source model to identify illicit copies. However, existing black-box fingerprinting methods often fail to generate distinctive LLM fingerprints. This ineffectiveness arises because black-box methods typically rely on model outputs, which lose critical information about the model's unique parameters due to the usage of non-linear functions. To address this, we first leverage Fisher Information Theory to formally demonstrate that the gradient of the model's input is a more informative feature for fingerprinting than the output. Based on this insight, we propose \name, a novel method that approximates these information-rich gradients in a black-box setting using zeroth-order estimation. \name overcomes the challenge of applying this to discrete text by simulating input perturbations via semantic-preserving word substitutions. This operation allows \name to estimate the model's Jacobian matrix as a unique fingerprint. Experiments on the standard benchmark show \name achieves a state-of-the-art effectiveness and robustness, significantly outperforming existing black-box methods. 
\end{abstract}


\keywords{Model Fingerprinting; AI Copyright Protection; Large Language Model; Trustworthy AI}


\maketitle
\newcommand\webconfavailabilityurl{https://doi.org/10.5281/zenodo.18325480}
\ifdefempty{\webconfavailabilityurl}{}{
\begingroup\small\noindent\raggedright\textbf{Resource Availability:}\\
The source code of this paper has been made publicly available at \url{\webconfavailabilityurl} and \url{https://github.com/shaoshuo-ss/ZeroPrint}.
\endgroup
}

\section{Introduction}
\label{sec:intro}

Large Language Models (LLMs) have become fundamental in revolutionizing the digital landscape~\cite{openai2023gpt, guo2025deepseek}, powering a vast array of web applications from intelligent search engines~\cite{sancheti2024llm, xiong2024search} and conversational agents~\cite{deng2024large, ning2025survey} to sophisticated content creation platforms~\cite{comanici2025gemini}. However, the development of these state-of-the-art (SOTA) models demands substantial investment in computational resources, massive datasets, and extensive research expertise, rendering them highly valuable intellectual property for their developers~\cite{wang2021riga, zeng2024huref, zhang2025reef}. Consequently, this value has given rise to significant security concerns regarding their copyright. Adversaries may engage in copyright infringement by using open-source models for commercial purposes without authorization~\cite{shao2025sok}, or even resort to model stealing~\cite{carlini2024stealing, yao2025black}, where proprietary models may be illicitly replicated and deployed. Protecting the copyright of these models is therefore not just a legal formality but a critical necessity to promote sustainable innovation and maintain a trustworthy AI ecosystem~\cite{xu2025copyright, ren2024sok}.


To address copyright concerns, LLM fingerprinting has emerged as a promising solution~\cite{pasquini2025llmmap, zeng2024huref, zhang2025reef, gubri2024trap}. This technique aims to determine whether a third-party suspicious model is derived from a copyrighted source model by extracting its intrinsic and unique features (\ie, a fingerprint) in a non-invasive manner~\cite{mcgovern2025your, zheng2022dnn, chen2022copy}. Copyright attribution is then performed by measuring the similarity between the fingerprints of the two models. Compared to invasive LLM watermarking techniques~\cite{li2023protecting, xu2024instructional, shao2025explanation}, which proactively embed a verifiable signal into a model's parameters or outputs before its release, fingerprinting offers significant advantages~\cite{shao2025sok}. First, it does not require model fine-tuning, which means it has a smaller overhead and also does not degrade the model's performance. Second, it can be applied to protect models that have already been deployed or released, a scenario where watermarking is not a viable option. Consequently, LLM fingerprinting has become a more practical and mainstream solution for copyright auditing.

Existing LLM fingerprinting techniques are primarily classified based on the level of access to the suspicious model, leading to two categories: white-box and black-box methods~\cite{shao2025sok}. White-box fingerprinting assumes full access to a model's internal components, including its architecture and parameters. This allows fingerprints to be extracted directly from static model weights~\cite{zeng2024huref, zheng2022dnn}, intermediate representations~\cite{zhang2025reef, zhang2024easydetector}, or gradients of weights~\cite{wu2025gradient}. Conversely, black-box fingerprinting operates under the constraint that the model developer can only interact with the model through the API and observe the corresponding outputs. This paradigm is further divided into two strategies. The first, \emph{untargeted fingerprinting}~\cite{sun2025idiosyncrasies, pasquini2025llmmap, jin2024proflingo}, involves submitting a selected set of queries to different models and comparing their outputs, where high similarity can indicate copyright infringement. The second strategy, \emph{targeted fingerprinting}~\cite{tsai2025rofl, gubri2024trap}, generates a unique set of query-response pairs by optimizing them on a source model and then verifies if a suspicious model can consistently reproduce these specific outputs.

However, a significant performance gap persists between white-box and black-box fingerprinting methods. A recent study~\cite{shao2025sok} reveals that white-box techniques, especially static methods that extract features directly from model parameters, achieve near-perfect auditing performance. This effectiveness likely stems from the immense scale of LLM parameters, where the distinct training paths of independently developed models cause their static weights to become inherently unique identifiers. 
The success of static fingerprinting thus suggests that \emph{model parameters may be the most critical features for reliable auditing}. A key challenge for black-box fingerprinting is therefore how to extract features that retain as much information about these parameters as possible. Existing black-box methods primarily rely on comparing model outputs. However, the extensive use of non-linear functions (\eg, GELU and SiLU~\cite{hendrycks2016gaussian}) throughout an LLM's architecture progressively compresses parameter information, leading to significant information loss in the final output. This raises a critical research question for the field: 

\emph{Is there another feature, accessible under black-box conditions, that can preserve more information about model parameters than output?}

To address this question, we first ground our investigation in a formal analysis of feature information content based on Fisher Information Theory~\cite{ly2017tutorial}. Our theoretical analysis reveals a crucial insight: given the input $X$, the amount of information about the parameter $W$ that can be obtained by observing the gradient $D=\mathrm{d}Y/\mathrm{d}X$ is approximately always greater than that by observing the output $Y$. This finding suggests that the gradient is an inherently more distinctive feature for a reliable fingerprint. Based on this understanding, we propose \name, a novel LLM fingerprinting method that captures this information-rich gradient in a black-box setting. The core challenge is that gradients are not directly accessible in the black-box auditing scenario. To overcome this, we adapt the principles of zeroth-order optimization~\cite{chen2017zoo}, a technique designed to approximate gradients using only input-output queries. However, classical zeroth-order methods require applying small, continuous perturbations to inputs, which is a process not directly applicable to the discrete and symbolic nature of text. \name tackles this gap by devising a method to simulate these perturbations in the text domain by generating semantically preserved variations of base queries through word substitution. By submitting these base and perturbed queries to a model, we can observe the corresponding changes in its output embeddings. These input-output difference pairs then allow us to estimate the local Jacobian matrix of the model using a regression-based approach, which serves as its unique fingerprint. By comparing the Jacobian fingerprint of a suspicious model with that of a source model, we can reliably determine its provenance.

Our contributions can be summarized as follows.
\begin{itemize}
    \item We provide a formal theoretical analysis grounded in Fisher Information Theory. Our analysis mathematically demonstrates that a model's gradient with respect to its input contains more information about the model's parameters than its output does, establishing gradients as a more robust feature for fingerprinting.
    \item We introduce \name, a novel black-box fingerprinting framework that estimates gradients as fingerprints without direct access to the model. \name overcomes the challenge of operating in a discrete text domain by simulating input perturbations through semantic-preserving word substitutions, enabling the approximation of the Jacobian matrix as a distinctive fingerprint.
    \item We conduct extensive experiments on \bench, a representative benchmark for LLM copyright auditing. The results show that \name consistently outperforms existing SOTA black-box fingerprinting methods across various metrics, demonstrating its superior effectiveness and reliability.
\end{itemize}

\section{Background}

\subsection{Large Language Models}

Large Language Models (LLMs) have become a cornerstone of the modern web, fundamentally reshaping how users interact with online services and information~\cite{deng2024large, zhu2025collaborative, luo2026shadow}. An LLM is a sophisticated computational system, typically defined by its vast parameter scale~\cite{guo2025deepseek, qwen2.5, grattafiori2024llama3}. The generation of text by these models follows an auto-regressive paradigm. This means that the model produces text sequentially, one token at a time, where each new token is conditionally dependent on the sequence of previously generated tokens. The probability of generating a sequence of tokens $Y=(y_1, y_2, \dots, y_T)$ is factorized into a product of conditional probabilities:
\begin{equation}
    P(Y)=\Pi_{t=1}^T P(y_t|y_1, y_2, \dots, y_{t-1}).
\end{equation}

This underlying mechanism enables the diverse web applications, such as search engines~\cite{sancheti2024llm}, code generation~\cite{nam2024using, he2025benchmarking}, and web agents~\cite{ning2025survey, deng2024large}. LLMs represent one of the most significant advancements in computing and have become immensely valuable digital assets. Consequently, protecting the copyright of these models within the web environment has emerged as a critical challenge.

\subsection{LLM Fingerprinting}


As LLMs are integrated into a growing number of applications, protecting the intellectual property of these models has become a critical task~\cite{xu2025copyright, li2025rethinking}. LLM fingerprinting is a key technique that has emerged for this purpose, designed to verify a model's provenance by analyzing its inherent characteristics~\cite{shao2025sok}. In this process, a copyright auditor can extract a unique signature (\ie, a fingerprint) from a suspicious model and compare it against the fingerprint of the source model to determine if it is an illicit copy or an unauthorized derivative. LLM fingerprinting provides a reliable mechanism for model owners to audit and protect their digital assets. Broadly, existing LLM fingerprinting can be categorized into two types: white-box and black-box LLM fingerprinting~\cite{shao2025sok}.

\partitle{White-box LLM Fingerprinting} White-box fingerprinting methods operate under the assumption that the auditor has full access to a model's internal components, including its architecture and parameters. These techniques are categorized based on their feature source into three different methods: static, forward-pass, and backward-pass methods~\cite{shao2025sok}. Static fingerprinting directly analyzes a model's learned parameters as the primary identifier~\cite{zeng2024huref, zheng2022dnn}. Forward-pass methods leverage the dynamic intermediate representations generated as the model processes inputs~\cite{zhang2025reef, zhang2024easydetector}, while backward-pass techniques use the gradients produced during backpropagation as a unique signature~\cite{wu2025gradient}. Although effective, the practical application of these methods is limited to open-source models where model parameters have already been accessible.

\partitle{Black-box LLM Fingerprinting} In contrast, black-box fingerprinting methods are more broadly applicable, as they only require API-level access to query a model and observe its outputs. These approaches can be divided into two main categories: untargeted and targeted fingerprinting~\cite{shao2025sok}.
\begin{itemize}
    \item \textbf{Untargeted Fingerprinting:} This is a straightforward and intuitive approach. The core idea is based on the observation that different LLMs often exhibit unique stylistic or behavioral ``idiosyncrasies'' in their generated outputs~\cite{sun2025idiosyncrasies}. Untargeted methods first select a set of queries, then submit them to different models and analyze the characteristics of the responses to generate a fingerprint~\cite{pasquini2025llmmap, gao2025model, mcgovern2025your, ren2025cotsrf}. A high degree of similarity between the fingerprints of a source and a suspicious model indicates that one is likely a derivative of the other.
    \item \textbf{Targeted Fingerprinting:} This paradigm seeks to construct a specific set of query-response pairs that are unique to a source model and its derivatives~\cite{gubri2024trap, tsai2025rofl, jin2024proflingo}. The key idea is that these special pairs act as a fingerprint, and only models from the source's lineage will be able to consistently reproduce the predefined target responses. To perform an audit, the auditor submits these optimized queries to a suspicious model and checks if its outputs match the target responses. If a significant number of responses match, the model will be flagged as a likely derivative.
\end{itemize}

Despite the applicability of black-box methods, the empirical results from \cite{shao2025sok} highlight a significant performance gap between the two paradigms. White-box methods demonstrate remarkable effectiveness and reliability due to their direct access to the model's stable and unique parameter space. Conversely, black-box techniques are found to be less reliable, as they solely depend on the outputs that are inherently more stochastic and less stable. This discrepancy underscores that reliably inferring a model's lineage with only black-box access remains challenging and highlights the significance of designing more effective solutions.

\section{Motivation and Theoretical Analysis}

The remarkable effectiveness of white-box fingerprinting, especially static methods that achieve near-perfect auditing performance by directly analyzing model parameters~\cite{shao2025sok}, offers a crucial insight. This success suggests that model parameters may be the most reliable features for fingerprinting. The rationale is that the immense scale of LLM parameters causes the static weights of independently developed models to become inherently unique identifiers.

In the black-box scenario, therefore, the key challenge is to \emph{extract features that retain as much information about the model's parameters as possible}. In this section, we analyze two features observable under black-box conditions: the model output $Y$ and the gradient (or the differential) $D=\mathrm{d}Y/\mathrm{d}X$. Existing black-box methods primarily rely on the direct observation of model outputs. However, our theoretical analysis reveals that observing the gradient $D$ contains approximately more information about the parameter $W$ than the output $Y$. This finding indicates that $D$ is likely a more reliable feature for fingerprinting, forming the theoretical foundation for our \name.

\subsection{The Nonlinearity Bottleneck: A Motivation Example}

To illustrate why observing the output $Y$ may yield less information about the model parameter $W$, let us consider a simple, fundamental unit commonly found in neural networks and LLMs:
\begin{equation}
    Y = f(WX + K),
\end{equation}
where $X$ and $Y$ represent the input and output, respectively, while $W$ and $K$ are the weight and bias. The function $f$ is a non-linear activation function. An entire neural network, including LLM, is constructed by stacking a vast number of these non-linear units. In this example, the gradient of the output $Y$ with respect to the input $X$ can be expressed as:
\begin{equation}
    D = \frac{\mathrm{d}Y}{\mathrm{d}X} = Wf'(WX+K).
\end{equation}

By analyzing both $Y$ and $D$, a key difference emerges. If we observe the gradient $D$, for a fixed input $X$, the term $f'(WX+K)$ becomes a constant value. Consequently, analyzing the characteristics of $D$ is effectively equivalent to analyzing the characteristics of the parameter $W$. In contrast, if we solely observe the output $Y$, the parameter $W$ must pass through the non-linear function $f$. Typically, irreversible non-linear functions (\eg, GELU and SiLU~\cite{hendrycks2016gaussian}) could compress the information, which can lead to a loss of information regarding $W$.

This observation leads us to a critical hypothesis: \emph{Is taking the gradient $D$ as a fingerprint a superior choice for preserving information about the parameter $W$ compared to taking the output $Y$?} In the following section, we will analyze this hypothesis from the perspective of information theory.

\subsection{The Theoretical Edge: Why Gradients Carry More Information}

In this section, we provide a theoretical analysis of the hypothesis proposed in the preceding section. To formally quantify and compare the amount of information that an observable feature contains about a model parameter, we employ the concept of \emph{Fisher Information} from information theory~\cite{ly2017tutorial}. Mathematically, Fisher Information quantifies how much an observable random variable $X$ carries about an unknown parameter $\theta$. A larger Fisher information value implies more information, which can in turn lead to a more precise estimate. We begin with the formal definition:

\begin{definition}[Fisher Information]
Let $\{p_{\theta}(x)\}_{\theta \in \Theta}$ be a family of probability density functions parameterized by $\theta$. Assume $p_{\theta}(x)$ is differentiable with respect to $\theta$ and satisfies the regularity conditions that allow interchange of differentiation and integration. The \textbf{Fisher information} of a single observation $X$ about the parameter $\theta$ is defined as
\begin{equation}
\mathcal{I}_X(\theta) 
:= 
\mathbb{E}_{X \sim p_{\theta}}
\left[
\left(
\frac{\partial}{\partial \theta} \log p_{\theta}(X)
\right)^{2}
\right],
\end{equation}
where the expectation is taken with respect to $p_{\theta}(x)$.  
\end{definition}

Based on this definition, to analyze whether observation $D$ or observation $Y$ contains more information about the model parameter $W$, it is equivalent to comparing their respective Fisher Information with regard to the parameter $W$. Through this analysis, we arrive at the following theorem:


\begin{theorem}
    \label{theorem:fisher}
    Let $Y = f(WX + K)$, where $f$ is neither a linear nor an affine function, $W$ and $K$ are constants with $W \neq 0$, and $D=\mathrm{d}Y/\mathrm{d}X=Wf'(WX+K)$. Suppose that $X$ follows a zero-mean normal distribution, denoted as $X \sim \mathcal{N}(0, \sigma_X^2)$, and $f''(K)\neq 0$. Under the first-order Taylor approximation, the Fisher information of $Y$ and $D$ with respect to $W$ satisfies:
    \begin{equation}
        \mathcal{I}_D(W) \geq (\frac{c_1^2}{2W^2c_2^2\sigma_X^2}+4) \cdot \mathcal{I}_Y(W),
    \end{equation}
    where $c_1=f'(K), c_2=f''(K)$.
\end{theorem}


The proof is in Appendix~\ref{sec:proof}. Theorem~\ref{theorem:fisher} provides the mathematical proof for our hypothesis. The coefficient $c_1^2/(2W^2c_2^2\sigma_X^2)+4 > 1$, leading to $\mathcal{I}_D(W) > \mathcal{I}_Y(W)$. The result formally establishes that the gradient $D$ contains approximately more information about the parameter $W$ than the output $Y$ does. This theoretical advantage confirms that the gradient is an inherently more distinctive and reliable feature for a robust model fingerprint.

\section{Methodology}

The theoretical analysis in the previous section provides strong evidence that gradients are more informative features for LLM fingerprinting than outputs. However, in a practical black-box scenario, we lack the direct access required for backpropagation to compute these gradients precisely. To bridge this gap, we introduce \name, a novel black-box LLM fingerprinting method that leverages zeroth-order gradient estimation to approximate this information-rich feature. In the following sections, we will first define the threat model under which our method operates and then present a detailed introduction to the \name framework.

\subsection{Threat Model}

In this paper, we consider a typical and strict threat model for black-box LLM fingerprinting, which involves two primary parties: the copyright auditor, who aims to protect their intellectual property, and a potential adversary, who seeks to infringe upon it.

\partitle{Adversary's Assumptions} The adversary's primary motivation is to leverage a SOTA, copyrighted model to develop their own service with minimal investment. We assume the adversary obtains an illicit copy or an unauthorized derivative of a source model and deploys it for commercial use. The adversary possesses full control over this suspicious model and may perform modifications (\eg, fine-tuning on a custom dataset or altering its behavior with different system prompts) to either adapt it for a specific task or attempt to obfuscate its origin.

\partitle{Auditor's Assumptions} The auditor operates under a set of constraints typical of a realistic black-box setting. We assume the auditor has the following capabilities:
\begin{itemize}
    \item The auditor has only black-box, API-level access to both the source and suspicious models. This access is limited to providing inputs and observing the corresponding textual outputs. Internal states and output logits are inaccessible in our threat model.
    \item The auditor's interactions with the source and suspicious LLMs are bound by a query limit to ensure practicality and mimic real-world API usage constraints. We assume the auditor can send a maximum of $q$ queries to a given LLM. 
\end{itemize}

\begin{figure*}
    \centering
    \includegraphics[width=0.99\linewidth]{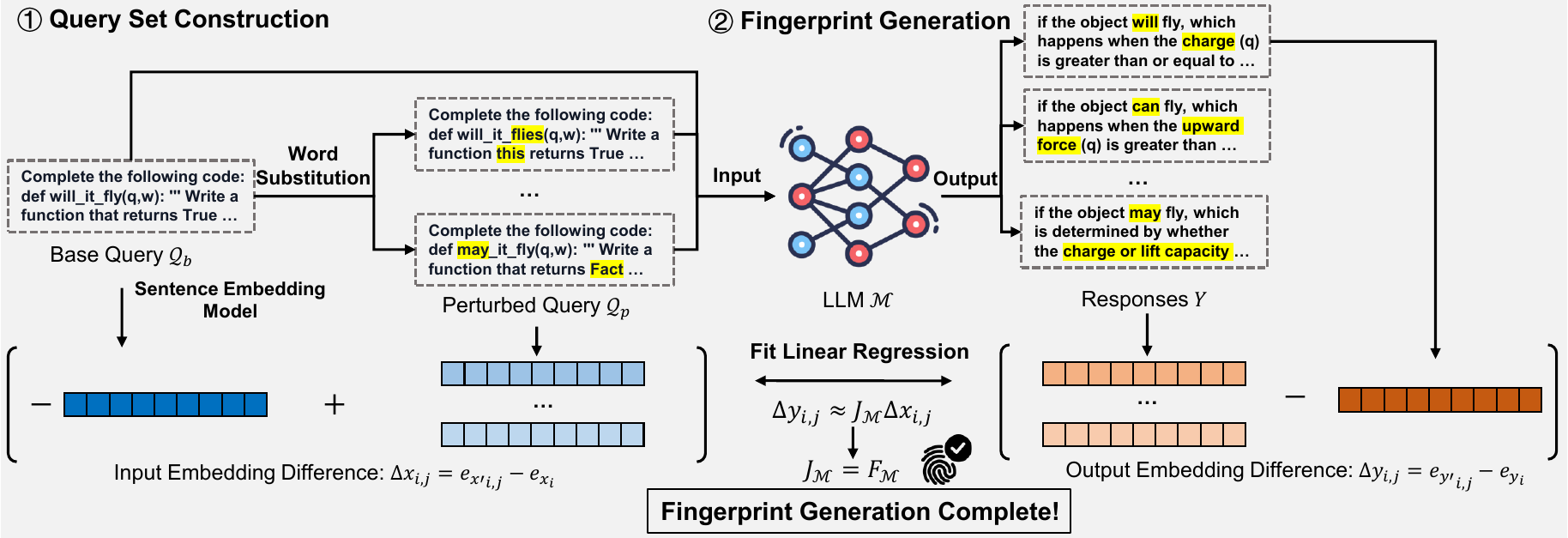}
    \vspace{-1em}
    \caption{The pipeline of \name. The process starts with creating a query set by generating semantically similar ``Perturbed Queries'' from a ``Base Query'' through word substitution. Both query types are fed into the LLM to get responses. All inputs and outputs are converted to embeddings to calculate input and output difference vectors. Finally, a linear regression model is fitted on these differences to estimate the Jacobian matrix, which serves as the model's unique fingerprint.}
    \label{fig:pipeline}
    \vspace{-1em}
\end{figure*}

\subsection{Overview of \name}

Our \name is structured as a three-step process, designed to effectively generate and compare model fingerprints in a black-box setting. The overall workflow is shown in Figure~\ref{fig:pipeline}, as follows:
\begin{itemize}
    \item \textbf{Step 1: Query Set Construction.} The process begins with the construction of a query set, denoted by $\mathcal{Q}$. This set comprises $N$ input prompts $\{x_1, x_2, \dots, x_N\}$ that are designed to elicit distinctive responses from the models being audited.
    \item \textbf{Step 2: Fingerprint Generation.} The query set $\mathcal{Q}$ is then fed into an LLM $\mathcal{M}$ to obtain a corresponding set of outputs, $Y$. Using these input-output pairs, we employ a zeroth-order gradient estimation technique to approximate the model's Jacobian matrix for each input. The (aggregated) estimated Jacobian matrix, $J_\mathcal{M}$, serves as the unique fingerprint of the LLM. Given the source model $\mathcal{M}_o$ and the suspicious model $\mathcal{M}_s$, we extract their fingerprints $F_o, F_s$ as follows:
    \begin{equation}
        F=\text{Extract}(\mathcal{M}(\mathcal{Q}), \mathcal{Q})\triangleq J\approx \frac{\partial \mathcal{M}(\mathcal{Q})}{\partial\mathcal{Q}},
    \end{equation}
    where $\mathcal{M}$ could be either $\mathcal{M}_o$ or $\mathcal{M}_s$.
    \item \textbf{Step 3: Fingerprint Comparison.} Finally, the similarity between the fingerprint of the source model, $J_o$, and that of the suspicious model, $J_s$, is calculated to determine provenance. In our proposed \name, we utilize the \emph{Pearson correlation coefficient}, which can be interpreted as a form of normalized cosine similarity, as the primary metric. If the similarity score exceeds a predefined threshold $\tau$, the suspicious model is flagged as a derivative of the source model.
    \begin{equation}
    \label{eq:sim}
        \text{Sim}(J_o, J_s)=\text{Pearson}(J_o, J_s).
    \end{equation}
\end{itemize}

The core design of the \name method hinges on two critical components: the construction of an effective query set and the specific algorithm used for fingerprint generation. A central challenge in this process is adapting gradient estimation for a black-box setting where inputs and outputs are discrete text.


In a conventional continuous domain, zeroth-order optimization estimates the gradient of a function $f: \mathbb{R}^\zeta \rightarrow \mathbb{R}^\eta$ by observing its output values. This is typically done by sampling a set of $m$ random direction vectors $\{u_i\}_{i=1}^{m}$ and querying the function at points perturbed along these directions, $f(x+\delta u_i)$, where $\delta$ is a small smoothing parameter. The Jacobian matrix $J$ (\ie, gradient) can then be approximated by solving a least-squares problem:
\begin{equation}
    \min_J \sum_{i=1}^{m} \| (f(x+\delta u_i) - f(x)) - J(\delta u_i)\|_2^2.
\end{equation}

\begin{table*}[t!]
\centering
\tabcolsep=4.5mm
\renewcommand{\arraystretch}{1.05}
\caption{Overall performance comparison between \name and baseline LLM fingerprinting methods.}
\label{tab:overall}
\vspace{-1em}
\scalebox{0.86}{
\begin{tabular}{l c c c c c c}
\toprule
 & \multicolumn{1}{c}{\textbf{White-box}} & \multicolumn{5}{c}{\textbf{Black-box}} \\
\cmidrule(lr){2-2} \cmidrule(lr){3-7}
\textbf{Metric} & REEF & LLMmap & MET & SEF & TRAP & \textbf{\name (Ours)} \\
\midrule
\textbf{AUC} $\uparrow$ & 0.896($\pm$0.021) & 0.632($\pm$0.009) & 0.661($\pm$0.003) & 0.581($\pm$0.008) & 0.712($\pm$0.027) & \textbf{0.720}($\pm$0.013)  \\
\textbf{pAUC} $\uparrow$ & 0.832($\pm$0.019) & 0.635($\pm$0.009) & 0.658($\pm$0.003) & 0.646($\pm$0.004) & 0.625($\pm$0.006) & \textbf{0.683}($\pm$0.005) \\
\textbf{TPR@1\%FPR} $\uparrow$ & 0.634($\pm$0.021) & 0.253($\pm$0.032) & 0.329($\pm$0.004) & 0.282($\pm$0.012) & 0.227($\pm$0.020) & \textbf{0.366}($\pm$0.009) \\
\textbf{MD} $\uparrow$ & 2.091($\pm$0.058) & 0.521($\pm$0.021) & \textbf{1.457}($\pm$0.004) & 0.172($\pm$0.014) & 1.382($\pm$0.064) & \textbf{1.457}($\pm$0.015) \\
\bottomrule
\end{tabular}
\vspace{-2em}
}
\end{table*}

However, this approach is not directly applicable in our black-box scenario. The input and output spaces are discrete and symbolic (\ie, text), not continuous vector spaces. It is not meaningful to add a small random vector to an input sentence. To overcome this, instead of adding continuous noise, \name generates perturbed inputs by performing semantic-preserving word substitutions on the base queries. We elaborate on the details in the following subsections.

\subsection{Query Set Construction}

The process of constructing the query set $\mathcal{Q}$ is composed of two main substeps. First, we select $n$ \emph{base queries} $\mathcal{Q}_b = \{x_1, \dots, x_n\}$ designed for stability across a wide range of models. Second, for each base query $x_i$, we generate a corresponding set of $m$ \emph{perturbed queries} $\mathcal{Q}_p^{(i)}=\{x'_{i, 1}, \dots, x'_{i,m}\}$ using a semantic-preserving word substitution technique, which is essential for the subsequent gradient estimation step. In total, we generate $n \times m + n = N$ samples as the Query Set $\mathcal{Q}$.

\partitle{Base Query Selection} A key consideration in designing a generalizable fingerprinting method is the varying instruction-following capabilities across different LLM instances, especially between pre-trained base models and their fine-tuned derivatives. Base models, for instance, often struggle to produce coherent replies to complex queries, such as those used in QA tasks. To ensure our method elicits consistent and meaningful behavior from diverse LLMs, the selected queries could align with the fundamental objective of their pre-training, which is essentially a completion task. Therefore, we utilize a code completion task to form our base query set. This approach is broadly compatible with models regardless of their fine-tuning stage. Specifically, we extract code snippets from a standard code generation dataset (\eg, HumanEval~\cite{chen2021humaneval}) and prepend them with a simple instruction: ``Complete the following code: [code snippet]''. From the dataset, we sample $n$ such code snippets to create our set of base queries.

\partitle{Constructing Perturbed Queries via Word Substitution} With the set of base queries, the next step is to generate perturbed samples for each one, which will be used to estimate the Jacobian. For each base query, we generate $n_p$ perturbed counterparts using a word substitution method grounded in semantic similarity.
The process leverages a pre-trained word embedding model (\eg, GloVe~\cite{pennington2014glove}) to find suitable replacements. For a given base query, we begin by randomly selecting $r$ replaceable words within the base query. For each of these $r$ words, we query the word embedding model to retrieve the top-$k$ most semantically similar words. A new, perturbed query is then formed by randomly replacing the original word with one of these $k$ candidates. This procedure is repeated independently to generate $m$ unique perturbed samples for each base query, resulting in a total of $n \times m$ perturbed queries. The example and the pseudocode can be found in Appendix~\ref{apd:details}.

\subsection{Fingerprint Generation via Zeroth-order Gradient Estimation}

Once the query set is constructed, the next step is to generate the LLM's fingerprint. This process involves three substeps: \textbf{(1)} querying the target LLM, \textbf{(2)} embedding the inputs and outputs into a continuous vector space, and \textbf{(3)} applying a regression method to estimate the Jacobian matrix, which serves as the final fingerprint. 

\partitle{Querying the Target LLM} The procedure begins by feeding the base queries $\{x_i\}_{i=1}^n$ and their corresponding perturbed counterparts $\{x'_{i,j}\}_{j=1}^m$ to the target LLM, $\mathcal{M}$. To mitigate the inherent randomness of LLM outputs, each query is sent to the model $t$ times, yielding a set of textual responses. 

\partitle{Embedding Inputs and Outputs} We then employ a pre-trained sentence embedding model to transform all inputs and outputs into high-dimensional vectors. For each query, the final output embedding is computed by averaging the embeddings of the $t$ responses, resulting in a stable representation. Let $E(\cdot)$ be the sentence embedding function. The embedded vector for the base input $x_i$ is $e_{x_i} = E(x_i)$, and its averaged output vector is $\bar{e}_{y_i}$. Similarly, the vectors for a perturbed query are $e_{x'_{i,j}} = E(x'_{i,j})$ and $\bar{e}_{y'_{i,j}}$.

\partitle{Estimate the Jacobian Matrix} With these vectors, we can formulate the Jacobian estimation as a linear regression problem. For each base query $x_i$, we compute the input and output difference vectors for all $m$ of its perturbations:
\begin{equation}
    \begin{aligned}
        &\text{Input difference:} &\Delta x_{i,j} = e_{x'_{i,j}} - e_{x_i}, \\
        &\text{Output difference:} &\Delta y_{i,j} = \bar{e}_{y'_{i,j}} - \bar{e}_{y_i}.
    \end{aligned}
\end{equation}

We seek to find a local Jacobian matrix $J_i$ that best models the linear relationship $\Delta y_{i,j} \approx J_i \Delta x_{i,j}$. Given the practical constraints on the number of queries allowed, the number of perturbation samples $m$ is often small relative to the dimensionality of the Jacobian matrix. This can lead to an underdetermined system, making an ordinary least squares estimate unstable. To address this, we use Ridge Regression, which introduces an L2 regularization term to produce a more robust estimate of the Jacobian:
\begin{equation}
    J_i = \arg\min_{J} \sum_{j=1}^{m} \| \Delta y_{i,j} - J \Delta x_{i,j} \|_2^2 + \alpha \| J \|_F^2,
\end{equation}
where $\alpha$ is the regularization hyperparameter and $\| \cdot \|_F$ is the Frobenius norm. Finally, a Jacobian matrix $J_i$ is computed for each of the $n$ base queries. The overall fingerprint for the model $\mathcal{M}$, denoted as $J_\mathcal{M}$, is the element-wise average of these individual Jacobian matrices. This aggregation step creates a more comprehensive and stable fingerprint that captures the model's characteristic behavior across a diverse set of inputs: $J_\mathcal{M} = \frac{1}{n} \sum_{i=1}^{n} J_i$.

With the final fingerprints generated for the source and suspicious models, the comparison is performed by calculating their similarity using Eq.~\eqref{eq:sim} to determine copyright attribution.

\section{Evaluation}

\subsection{Experimental Settings}

\partitle{Benchmark Models} The primary benchmark models in our experiments are from \bench~\cite{shao2025sok}. \bench is built upon 7 mainstream pre-trained foundation models, including variants of Qwen~\cite{qwen2.5}, Llama~\cite{grattafiori2024llama3, touvron2023llama2}, Mistral~\cite{mistral}, Gemma~\cite{gemma2}, and TinyLlama~\cite{zhang2024tinyllama}. These models were selected with parameter counts ranging from 1.1B to 14B. \bench also implements various LLM post-development techniques, including fine-tuning, model merging, distillation, different system prompts, RAG, adversarial attacks, etc. In total, \bench provides 149 different model instances and establishes a representative benchmark for evaluating fingerprint-based LLM copyright auditing. The details of \bench can be found in Appendix~\ref{apd:details}.

\partitle{Baselines} For the baseline methods, we implement 3 untargeted black-box methods, LLMmap~\cite{pasquini2025llmmap}, Model Equality Testing (MET)~\cite{gao2025model}, and Sentence Embedding Fingerprinting (SEF)~\cite{shao2025sok}, and 1 targeted black-box method, TRAP~\cite{gubri2024trap}. We also include a white-box method, REEF~\cite{zhang2025reef}, for reference. Note that targeted methods, including TRAP, are typically \emph{semi-black-box}. They assume black-box access to the suspicious model, but white-box access to the source model, which is a slightly stronger assumption than untargeted methods. The maximum number of queries is set to $q=200$.

\partitle{Settings of \name} For default settings, we select $n=2$ base queries from HumanEval~\cite{chen2021humaneval} and generate $m=4$ perturbed queries, with a total number of $10$ queries. We input each query to the LLM $t=20$ times, leading to $200$ queries in total. For the sentence embedding model, we choose all-mpnet-base-v2~\cite{reimers2019sentence}, which is a classic yet powerful model also used in LLMmap. More detailed settings are in Appendix~\ref{apd:details}.

\partitle{Metrics} In our experiments, following prior work~\cite{shao2025sok}, we primarily leverage two metrics: \textbf{(1)} Area Under the ROC Curve (\textbf{AUC}) and \textbf{(2)} Partial AUC with FPR$\in[0, 0.05]$ (\textbf{pAUC}). The AUC provides a holistic measure of a method's ability to distinguish between related model pairs (a source model and its derivative) and unrelated pairs across all possible classification thresholds $\tau$. In contrast, the pAUC metric is specifically designed to assess performance by focusing on a specific low false-positive-rate (FPR) region (\ie, $[0, 0.05]$) of the ROC curve. Higher AUC and pAUC scores indicate superior performance of a fingerprinting method.

\begin{figure*}
    \centering
    \includegraphics[width=0.85\linewidth]{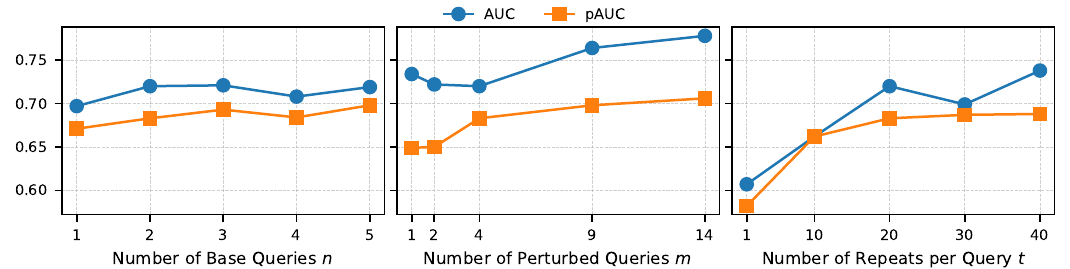}
    \vspace{-1em}
    \caption{The performance of \name with different numbers of queries. We study three factors: the number of base queries $n$, the number of perturbed queries $m$, and the number of repeats per query $t$.}
    \label{fig:querynumber}
    \vspace{-0.5em}
\end{figure*}

\subsection{Main Results}

Table~\ref{tab:overall} presents the overall performance comparison between \name and various baseline methods. In addition to the primary metrics of AUC and pAUC, we also report the True Positive Rate at a 1\% False Positive Rate (TPR@1\%FPR) and the Mahalanobis Distance (MD). The TPR@1\%FPR metric measures the true positive rate that a method can achieve while maintaining a low FPR of 1\%. The MD metric measures the separability of the fingerprint distributions, reflecting the distance between positive pairs (a source model and its derivatives) and negative pairs (independent models). A larger MD indicates a stronger ability to distinguish between related and unrelated models. For both TPR@1\%FPR and MD, higher values signify better performance.

The results demonstrate that \name achieves SOTA or comparable performance across all evaluated metrics. It significantly outperforms existing untargeted methods. Notably, \name also shows superior performance compared to the targeted, semi-black-box method TRAP, improving the AUC to 0.720 and the pAUC to 0.683. However, a performance gap remains when compared to white-box methods like REEF, which have direct access to model internals. It is worth noting that these results were achieved under a strict query limit of 200 queries per model. Subsequent experiments in our ablation study reveal that \name's performance can be further enhanced by increasing the number of queries allowed. This scalability is further detailed in Section~\ref{sec:ablation}.

\subsection{Ablation Study}
\label{sec:ablation}


\partitle{Effects of Different Numbers of Queries} The total number of queries is a critical factor influencing the performance of \name. This number is determined by three key hyperparameters: the number of base queries ($n$), the number of perturbed queries generated for each base query ($m$), and the number of repeated generations for each query ($t$). To understand their impact, we conducted a series of experiments, varying one parameter at a time while keeping others fixed. The results are illustrated in Figure~\ref{fig:querynumber}.
\begin{itemize}
    \item \textbf{Effects of Different Numbers of Base Queries $n$.} We first investigated the impact of $n$ by setting its value to 1, 2, 3, 4, and 5, with $n=2$ being the default in our main experiments. As shown in the first subplot of Figure~\ref{fig:querynumber}, the performance of \name remains relatively stable across different values of $n$. The AUC and pAUC fluctuate in a small interval. This observation suggests that the number of base queries may not have a significant impact on the effectiveness of \name. A small number of base queries (\eg, 2) is sufficient to generate a distinctive fingerprint.
    \item \textbf{Effects of Different Numbers of Perturbed Queries $m$.} Next, we examined the influence of $m$ with values set to 1, 2, 4, 9, and 14. The default value used in our main experiments was $m=4$. The results, depicted in the second subplot of Figure~\ref{fig:querynumber}, reveal that $m$ has a pronounced effect on the method's performance. When $m$ is small (\eg, 1 or 2), the estimated gradients are less stable, leading to inconsistent results (\eg, a high AUC but a low pAUC). However, as m increases to 4 and beyond, both AUC and pAUC show a clear and significant upward trend. The performance peaks at $m=14$, achieving the highest AUC (0.778) and pAUC (0.706) in this set of experiments. This trend highlights the potential of \name to achieve even greater effectiveness and robustness with an increased query budget, as more perturbed queries provide a more accurate gradient estimation.
    \item \textbf{Effects of Different Numbers of Repeats per Query $t$.} Finally, we analyzed the effect of $t$, testing values of 1, 10, 20, 30, and 40, with $t=20$ as the default setting. This parameter is designed to mitigate the inherent randomness in LLM-generated outputs. The third subplot of Figure~\ref{fig:querynumber} shows that when $t$ is small, performance suffers due to this stochasticity, yielding an AUC of only 0.607 for $t=1$. As $t$ increases, the performance improves substantially, with the AUC reaching 0.720 at $t=20$. Beyond this point, the results begin to plateau and fluctuate slightly. This indicates that at $t=20$, the averaged output embedding for a given input becomes sufficiently stable. Further increases in $t$ offer diminishing returns.
\end{itemize}

\begin{table}[t]
\centering
\tabcolsep=1.2mm
\renewcommand{\arraystretch}{1.05}
\caption{The performance of \name using base queries from different datasets.}
\label{tab:dataset}
\vspace{-1em}
\scalebox{0.85}{
\begin{tabular}{l c c c c c}
\toprule
\textbf{Dataset} & TruthfulQA & SQuAD & DROP & Wikipedia & \textbf{HumanEval (Ours)} \\
\midrule
\textbf{AUC} $\uparrow$ & 0.622 & 0.665 & 0.676 & 0.698 & \textbf{0.720}  \\
\textbf{pAUC} $\uparrow$ & 0.668 & 0.666 & 0.651 & 0.647 & \textbf{0.683} \\
\bottomrule
\end{tabular}
}
\vspace{-0.5em}
\end{table}

\partitle{Effects of Different Query Data} The choice of data for constructing base queries is another crucial element. We evaluated the performance of \name using base queries derived from five different datasets: HumanEval~\cite{chen2021humaneval} (our default), three question-answering (QA) datasets (TruthfulQA~\cite{lin2022truthfulqa}, SQuAD~\cite{rajpurkar2016squad}, and DROP~\cite{dua2019drop}), and one text completion dataset (Wikipedia~\cite{wikidump}). The results, presented in Table~\ref{tab:dataset}, show that the source of the query data significantly impacts fingerprinting effectiveness. When using queries from the QA datasets, the performance was consistently poor. Similarly, the Wikipedia dataset, despite being a completion task, also failed to achieve comparable results to HumanEval. This can be attributed to the nature of the tasks. Pure text completion is inherently more diverse and less structured than code completion. Thus, various derivative models tend to exhibit less consistent behavior on text-based prompts, making it harder to extract a stable and reliable fingerprint. In contrast, the structured nature of code completion elicits more uniform responses, leading to better performance.


\partitle{Effects of Different Sentence Embedding Models} We investigated the influence of the sentence embedding model on \name's performance. We experimented with five different models of varying scales and embedding dimensions (details in Appendix~\ref{apd:details}). The results, detailed in Table~\ref{tab:model}, reveal a clear trade-off related to the embedding dimension. When the dimension is too small, the model fails to adequately capture the semantic nuances of the sentences, leading to poor performance. Conversely, a larger embedding dimension increases the size of the Jacobian matrix to be estimated. Given the constraints on the number of queries, estimating this larger matrix becomes less accurate, which also degrades the overall fingerprinting effectiveness. This is supported by the similar performance of MPNet and EmbeddingGemma, which share the same embedding dimension and achieve comparable results.

\begin{table}[t!]
\centering
\tabcolsep=0.9mm
\renewcommand{\arraystretch}{1.0}
\caption{The performance of \name using different sentence embedding models. ``Dim'' signifies the dimension of the generated sentence embeddings.}
\label{tab:model}
\vspace{-1em}
\scalebox{0.84}{
\begin{tabular}{l c c c c c}
\toprule
\textbf{Dataset}  & Model2Vec & MiniLM & MPNet \textbf{(main)} & EmbeddingGemma & Qwen3 \\
\textbf{Dim}  & 256 & 384 & 768 & 768  & 1024  \\
\midrule
\textbf{AUC} $\uparrow$ & 0.683 & 0.702 & 0.720 & 0.721 & 0.698  \\
\textbf{pAUC} $\uparrow$ & 0.659 & 0.682 & 0.683  & 0.682 & 0.659 \\
\bottomrule
\end{tabular}
}
\vspace{-0.5em}
\end{table}

\subsection{Resistance to Potential Adaptive Attacks}

\begin{table}[t!]
\centering
\tabcolsep=2.5mm
\renewcommand{\arraystretch}{1.0}
\caption{The fingerprint similarities between the original instruct LLM and its attacked version using \name.}
\label{tab:attack}
\vspace{-1em}
\scalebox{0.85}{
\begin{tabular}{l c c c c c c}
\toprule
\textbf{Attack} & \multicolumn{3}{c}{\textbf{Input Paraphrasing}} & \multicolumn{3}{c}{\textbf{Output Perturbation}} \\
\cmidrule(lr){2-4} \cmidrule(lr){5-7}
\textbf{Setting} & 0.5B & 1B & 2B & 0.05 & 0.10 & 0.15  \\
\midrule
\textbf{Qwen2.5-7B} & 0.878 & 0.876 & 0.881 & 0.865 & 0.894 & 0.868  \\
\textbf{Llama3.1-8B} & 0.875 & 0.913 & 0.904  & 0.887 & 0.868 & 0.903 \\
\bottomrule
\end{tabular}
}
\vspace{-1em}
\end{table}

To assess the robustness of \name, we evaluated its resilience against two potential adaptive attack strategies an adversary might employ to evade detection:
\begin{itemize}
    \item \textbf{Input Paraphrasing Attack}: This attack attempts to obfuscate queries by rewriting them using a separate, smaller LLM before submitting them to the target model. We simulated this using three paraphrasing models of varying sizes (0.5B, 1B, and 2B).
    \item \textbf{Output Perturbation Attack}: This attack involves adding noise to the model's output logits to disrupt the fingerprinting process. We tested this by introducing Gaussian noise with three different standard deviations (0.05, 0.10, and 0.15).
\end{itemize}

Detailed settings can be found in Appendix~\ref{apd:details}. The experimental results, summarized in Table~\ref{tab:attack}, demonstrate that \name maintains high fingerprint similarity under both attack scenarios, confirming its strong resilience. The robustness against the input paraphrasing attack stems from our method's reliance on the semantic difference between embeddings. While paraphrasing alters the surface form of the text, the core semantic relationship and vector difference between the base and perturbed queries remain relatively stable, allowing for a reliable estimation of the Jacobian matrix. For the output perturbation attack, \name's resilience can be attributed to two key factors. First, averaging the embeddings from multiple queries effectively filters out random, zero-mean noise. Second, since our fingerprint is derived from the estimated gradient, which is the relationship between input and output changes, it is inherently less sensitive to small, random perturbations in the absolute output values, thus preserving the fingerprint's integrity.

\begin{figure}
    \centering
    \includegraphics[width=0.90\linewidth]{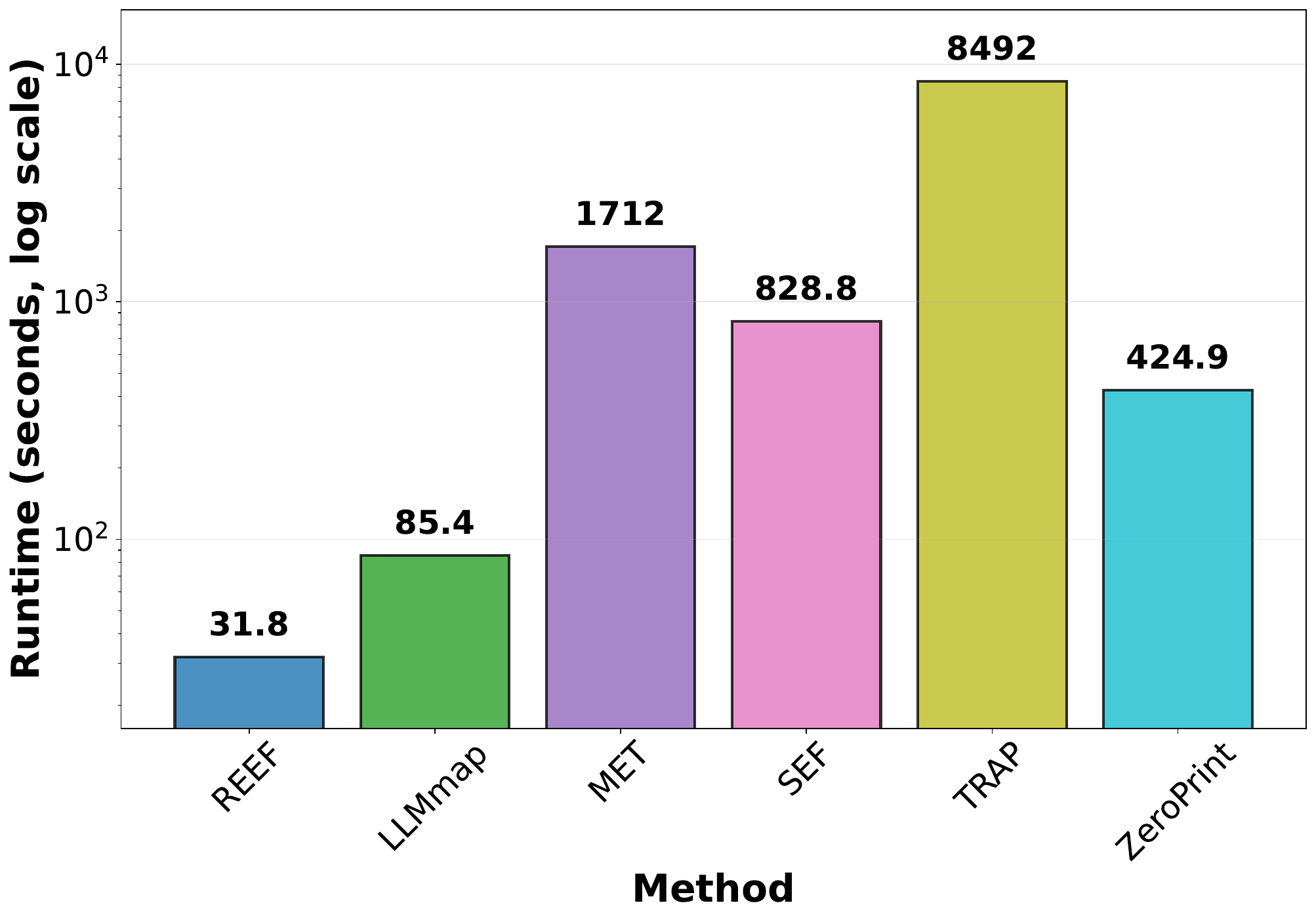}
    \vspace{-1em}
    \caption{The overhead of different fingerprinting methods.}
    \label{fig:overhead}
    \vspace{-1em}
\end{figure}

\subsection{Overhead Evaluation}
\label{sec:overhead}

To evaluate practical efficiency, we measured the total time for fingerprint extraction and verification on a subset of four LLMs. Experimental details can be found in Appendix~\ref{apd:details}.  From Figure~\ref{fig:overhead}, \name demonstrates a competitive overhead of 424.9 seconds, which is significantly faster than most black-box baselines like MET, SEF, and TRAP. Its runtime is only greater than the white-box REEF and LLMmap, the latter of which uses a minimal set of only 8 queries. Given that the fingerprinting for a single model can be completed in under two minutes on average, \name's process completes in an acceptable timeframe. Overall, \name is a computationally practical solution for real-world auditing.

\section{Conclusion}

In this paper, we addressed the critical challenge of protecting the intellectual property of large language models (LLMs) through fingerprinting, a technique that aims to identify a model's origin by extracting a unique, intrinsic signature. We began by formally demonstrating, through Fisher Information Theory, that a model's gradient contains significantly more information about its parameters than its output, establishing it as a superior feature for fingerprinting. Building on this theoretical foundation, we introduced \name, a novel black-box fingerprinting method that successfully approximates these information-rich gradients. By simulating input perturbations through semantic-preserving word substitutions, \name overcomes the inherent challenges of applying gradient estimation to the discrete domain of text. Our extensive experiments show that \name not only achieves SOTA performance, outperforming existing black-box methods in both effectiveness and robustness, but also operates with a practical overhead. This work represents a significant step forward in developing reliable and secure mechanisms for LLM copyright auditing.

\section*{Acknowledgments}

This research is supported in part by the ``Pioneer'' and ``Leading Goose'' R\&D Program of Zhejiang (2024C01169), the Kunpeng-Ascend Science and Education Innovation Excellence/Incubation Center, and the National Natural Science Foundation of China under Grants (62441238, U2441240). 


\bibliographystyle{ACM-Reference-Format}
\bibliography{ref}

@article{openai2023gpt,
  title={GPT-4 Technical Report},
  author={OpenAI},
  journal={arXiv preprint arXiv:2303.08774},
  year={2023}
}

@article{guo2025deepseek,
  title={Deepseek-r1: Incentivizing reasoning capability in llms via reinforcement learning},
  author={Guo, Daya and Yang, Dejian and Zhang, Haowei and Song, Junxiao and Zhang, Ruoyu and others},
  journal={Nature},
  pages={633–-638},
  volumn={645},
  year={2025}
}

@article{touvron2023llama2,
  title={Llama 2: Open foundation and fine-tuned chat models},
  author={Touvron, Hugo and Martin, Louis and Stone, Kevin and Albert, Peter and Almahairi, Amjad and Babaei, Yasmine and Bashlykov, Nikolay and Batra, Soumya and Bhargava, Prajjwal and Bhosale, Shruti and others},
  journal={arXiv preprint arXiv:2307.09288},
  year={2023}
}

@inproceedings{chen2022copy,
  title={Copy, right? a testing framework for copyright protection of deep learning models},
  author={Chen, Jialuo and Wang, Jingyi and Peng, Tinglan and Sun, Youcheng and Cheng, Peng and Ji, Shouling and Ma, Xingjun and Li, Bo and Song, Dawn},
  booktitle={IEEE Symposium on Security and Privacy},
  pages={824--841},
  year={2022},
}

@inproceedings{li2023protecting,
  title={Protecting intellectual property of large language model-based code generation apis via watermarks},
  author={Li, Zongjie and Wang, Chaozheng and Wang, Shuai and Gao, Cuiyun},
  booktitle={ACM SIGSAC Conference on Computer and Communications Security},
  pages={2336--2350},
  year={2023}
}

@article{ren2024sok,
  title={SoK: On the Role and Future of AIGC Watermarking in the Era of Gen-AI},
  author={Ren, Kui and Yang, Ziqi and Lu, Li and Liu, Jian and Li, Yiming and Wan, Jie and Zhao, Xiaodi and Feng, Xianheng and Shao, Shuo},
  journal={arXiv preprint arXiv:2411.11478},
  year={2024}
}

@inproceedings{shao2025explanation,
  title={Explanation as a Watermark: Towards Harmless and Multi-bit Model Ownership Verification via Watermarking Feature Attribution},
  author={Shao, Shuo and Li, Yiming and Yao, Hongwei and He, Yiling and Qin, Zhan and Ren, Kui},
  booktitle={Network and Distributed System Security Symposium},
  year={2025}
}

@inproceedings{he2025benchmarking,
  title={On Benchmarking Code LLMs for Android Malware Analysis},
  author={He, Yiling and She, Hongyu and Qian, Xingzhi and Zheng, Xinran and Chen, Zhuo and Qin, Zhan and Cavallaro, Lorenzo},
  booktitle={ACM SIGSOFT International Symposium on Software Testing and Analysis Workshop},
  year={2025}
}

@inproceedings{carlini2024stealing,
  title={Stealing part of a production language model},
  author={Carlini, Nicholas and Paleka, Daniel and Dvijotham, Krishnamurthy and Steinke, Thomas and Hayase, Jonathan and Cooper, A Feder and Lee, Katherine and Jagielski, Matthew and Nasr, Milad and Conmy, Arthur and others},
  booktitle={International Conference on Machine Learning},
  pages={5680--5705},
  year={2024}
}

@inproceedings{xu2024instructional,
  title={Instructional Fingerprinting of Large Language Models},
  author={Xu, Jiashu and Wang, Fei and Ma, Mingyu and Koh, Pang Wei and Xiao, Chaowei and Chen, Muhao},
  booktitle={Conference of the North American Chapter of the Association for Computational Linguistics},
  year={2024}
}

@article{grattafiori2024llama3,
  title={The llama 3 herd of models},
  author={Grattafiori, Aaron and Dubey, Abhimanyu and Jauhri, Abhinav and Pandey, Abhinav and Kadian, Abhishek and Al-Dahle, Ahmad and Letman, Aiesha and Mathur, Akhil and Schelten, Alan and Vaughan, Alex and others},
  journal={arXiv preprint arXiv:2407.21783},
  year={2024}
}

@inproceedings{pasquini2025llmmap,
  title={Llmmap: Fingerprinting for large language models},
  author={Pasquini, Dario and Kornaropoulos, Evgenios M and Ateniese, Giuseppe},
  booktitle={USENIX Security Symposium},
  year={2025}
}

@inproceedings{zhang2025reef,
  title={Reef: Representation encoding fingerprints for large language models},
  author={Zhang, Jie and Liu, Dongrui and Qian, Chen and Zhang, Linfeng and Liu, Yong and Qiao, Yu and Shao, Jing},
  booktitle={International Conference on Learning Representations},
  year={2025}
}

@inproceedings{zeng2024huref,
  title={Huref: Human-readable fingerprint for large language models},
  author={Zeng, Boyi and Wang, Lizheng and Hu, Yuncong and Xu, Yi and Zhou, Chenghu and Wang, Xinbing and Yu, Yu and Lin, Zhouhan},
  booktitle={Annual Conference on Neural Information Processing Systems},
  year={2024}
}

@article{li2025rethinking,
  title={Rethinking Data Protection in the (Generative) Artificial Intelligence Era}, 
  author={Li, Yiming and Shao, Shuo and He, Yu and Guo, Junfeng and Zhang, Tianwei and Qin, Zhan and Chen, Pin-Yu and Backes, Michael and Torr, Philip and Tao, Dacheng and Ren, Kui},
  year={2025},
  journal={arXiv preprint arXiv:2507.03034}
}

@inproceedings{zhang2024easydetector,
  title={EasyDetector: Using Linear Probe to Detect the Provenance of Large Language Models},
  author={Zhang, Jie and Li, Jiayuan and Fei, Haiqiang and Li, Lun and Zhu, Hongsong},
  booktitle={IEEE International Conference on Trust, Security and Privacy in Computing and Communications},
  year={2024},
}

@article{wu2025gradient,
  title={Gradient-Based Model Fingerprinting for LLM Similarity Detection and Family Classification},
  author={Wu, Zehao and Zhao, Yanjie and Wang, Haoyu},
  journal={arXiv preprint arXiv:2506.01631},
  year={2025}
}

@inproceedings{gubri2024trap,
  title={TRAP: Targeted Random Adversarial Prompt Honeypot for Black-Box Identification},
  author={Gubri, Martin and Ulmer, Dennis and Lee, Hwaran and Yun, Sangdoo and Oh, Seong Joon},
  booktitle={Findings of the Association for Computational Linguistics: ACL 2024},
  pages={11496--11517},
  year={2024}
}

@inproceedings{jin2024proflingo,
  title={Proflingo: A fingerprinting-based intellectual property protection scheme for large language models},
  author={Jin, Heng and Zhang, Chaoyu and Shi, Shanghao and Lou, Wenjing and Hou, Y Thomas},
  booktitle={IEEE Conference on Communications and Network Security},
  pages={1--9},
  year={2024},
  organization={IEEE}
}

@inproceedings{gao2025model,
  title={Model Equality Testing: Which Model Is This API Serving?},
  author={Gao, Irena and Liang, Percy and Guestrin, Carlos},
  booktitle={International Conference on Learning Representations},
  year={2025}
}

@inproceedings{mcgovern2025your,
  title={Your Large Language Models are Leaving Fingerprints},
  author={McGovern, Hope Elizabeth and Stureborg, Rickard and Suhara, Yoshi and Alikaniotis, Dimitris},
  booktitle={International Conference on Computational Linguistics Workshop},
  pages={85--95},
  year={2025}
}

@inproceedings{sun2025idiosyncrasies,
  title={Idiosyncrasies in large language models},
  author={Sun, Mingjie and Yin, Yida and Xu, Zhiqiu and Kolter, J Zico and Liu, Zhuang},
  booktitle={International Conference on Machine Learning},
  year={2025}
}

@article{tsai2025rofl,
  title={RoFL: Robust Fingerprinting of Language Models},
  author={Tsai, Yun-Yun and Guo, Chuan and Yang, Junfeng and van der Maaten, Laurens},
  journal={arXiv preprint arXiv:2505.12682},
  year={2025}
}

@article{ren2025cotsrf,
  title={CoTSRF: Utilize Chain of Thought as Stealthy and Robust Fingerprint of Large Language Models},
  author={Ren, Zhenzhen and Li, GuoBiao and Li, Sheng and Qian, Zhenxing and Zhang, Xinpeng},
  journal={arXiv preprint arXiv:2505.16785},
  year={2025}
}

@misc{qwen2.5,
    title = {Qwen2.5: A Party of Foundation Models},
    url = {https://qwenlm.github.io/blog/qwen2.5/},
    author = {Qwen Team},
    year = {2024}
}

@article{zheng2022dnn,
  title={A DNN fingerprint for non-repudiable model ownership identification and piracy detection},
  author={Zheng, Yue and Wang, Si and Chang, Chip-Hong},
  journal={IEEE Transactions on Information Forensics and Security},
  volume={17},
  pages={2977--2989},
  year={2022},
  publisher={IEEE}
}

@article{mistral,
  title={Mistral 7B},
  author={Jiang, Albert and Sablayrolles, Alexandre and Mensch, Arthur and Bamford, Chris and others},
  journal={arXiv preprint arXiv:2310.06825},
  year={2023}
}

@article{gemma2,
  title={Gemma 2: Improving open language models at a practical size},
  author={Gemma Team},
  journal={arXiv preprint arXiv:2408.00118},
  year={2024}
}

@article{zhang2024tinyllama,
  title={Tinyllama: An open-source small language model},
  author={Zhang, Peiyuan and Zeng, Guangtao and Wang, Tianduo and Lu, Wei},
  journal={arXiv preprint arXiv:2401.02385},
  year={2024}
}

@inproceedings{luo2026shadow,
  author = {Luo, Zhifan and Shao, Shuo and Zhang, Su and Zhou, Lijing and Hu, Yuke and Liu, Zhihao and Qin, Zhan},
  title = {Shadow in the Cache: Unveiling and Mitigating Privacy Risks of KV-cache in LLM Inference},
  booktitle = {Network and Distributed System Security Symposium},
  year = {2026}
}

@article{zou2023universal,
  title={Universal and transferable adversarial attacks on aligned language models},
  author={Zou, Andy and Wang, Zifan and Carlini, Nicholas and Nasr, Milad and Kolter, J Zico and Fredrikson, Matt},
  journal={arXiv preprint arXiv:2307.15043},
  year={2023}
}

@inproceedings{dua2019drop,
  title={DROP: A Reading Comprehension Benchmark Requiring Discrete Reasoning Over Paragraphs},
  author={Dua, Dheeru and Wang, Yizhong and Dasigi, Pradeep and Stanovsky, Gabriel and Singh, Sameer and Gardner, Matt},
  booktitle={Conference of the North American Chapter of the Association for Computational Linguistics: Human Language Technologies},
  year={2019}
}

@inproceedings{hendrycks2021ethics,
  title={Aligning AI With Shared Human Values},
  author={Dan Hendrycks and Collin Burns and Steven Basart and Andrew Critch and Jerry Li and Dawn Song and Jacob Steinhardt},
  booktitle={International Conference on Learning Representations},
  year={2021}
}

@inproceedings{jin2019pubmedqa,
  title={PubMedQA: A Dataset for Biomedical Research Question Answering},
  author={Jin, Qiao and Dhingra, Bhuwan and Liu, Zhengping and Cohen, William and Lu, Xinghua},
  booktitle={Conference on Empirical Methods in Natural Language Processing and International Joint Conference on Natural Language Processing},
  pages={2567--2577},
  year={2019}
}

@article{chen2021humaneval,
  title={Evaluating large language models trained on code},
  author={Chen, Mark and Tworek, Jerry and Jun, Heewoo and Yuan, Qiming and others},
  journal={arXiv preprint arXiv:2107.03374},
  year={2021}
}

@article{qwen3embedding,
  title={Qwen3 Embedding: Advancing Text Embedding and Reranking Through Foundation Models},
  author={Zhang, Yanzhao and Li, Mingxin and Long, Dingkun and Zhang, Xin and others},
  journal={arXiv preprint arXiv:2506.05176},
  year={2025}
}

@article{xu2025copyright,
  title={Copyright Protection for Large Language Models: A Survey of Methods, Challenges, and Trends},
  author={Xu, Zhenhua and Yue, Xubin and Wang, Zhebo and Liu, Qichen and Zhao, Xixiang and Zhang, Jingxuan and Zeng, Wenjun and Xing, Wengpeng and Kong, Dezhang and Lin, Changting and others},
  journal={arXiv preprint arXiv:2508.11548},
  year={2025}
}

@inproceedings{ning2025survey,
  title={A survey of webagents: Towards next-generation ai agents for web automation with large foundation models},
  author={Ning, Liangbo and Liang, Ziran and Jiang, Zhuohang and Qu, Haohao and Ding, Yujuan and Fan, Wenqi and Wei, Xiao-yong and Lin, Shanru and Liu, Hui and Yu, Philip S and others},
  booktitle={ACM SIGKDD Conference on Knowledge Discovery and Data Mining},
  pages={6140--6150},
  year={2025}
}

@inproceedings{sancheti2024llm,
  title={Llm driven web profile extraction for identical names},
  author={Sancheti, Prateek and Karlapalem, Kamalakar and Vemuri, Kavita},
  booktitle={The ACM Web Conference},
  year={2024}
}

@inproceedings{deng2024large,
  title={Large language model powered agents in the web},
  author={Deng, Yang and Zhang, An and Lin, Yankai and Chen, Xu and Wen, Ji-Rong and Chua, Tat-Seng},
  booktitle={The ACM Web Conference},
  pages={1242--1245},
  year={2024}
}

@article{shao2025sok,
  title={SoK: Large Language Model Copyright Auditing via Fingerprinting},
  author={Shao, Shuo and Li, Yiming and He, Yu and Yao, Hongwei and Yang, Wenyuan and Tao, Dacheng and Qin, Zhan},
  journal={arXiv preprint arXiv:2508.19843},
  year={2025}
}

@article{ly2017tutorial,
  title={A tutorial on Fisher information},
  author={Ly, Alexander and Marsman, Maarten and Verhagen, Josine and Grasman, Raoul PPP and Wagenmakers, Eric-Jan},
  journal={Journal of Mathematical Psychology},
  volume={80},
  pages={40--55},
  year={2017},
  publisher={Elsevier}
}

@inproceedings{wang2021riga,
  title={Riga: Covert and robust white-box watermarking of deep neural networks},
  author={Wang, Tianhao and Kerschbaum, Florian},
  booktitle={The ACM Web Conference},
  pages={993--1004},
  year={2021}
}

@article{comanici2025gemini,
  title={Gemini 2.5: Pushing the frontier with advanced reasoning, multimodality, long context, and next generation agentic capabilities},
  author={Comanici, Gheorghe and Bieber, Eric and Schaekermann, Mike and Pasupat, Ice and Sachdeva, Noveen and Dhillon, Inderjit and Blistein, Marcel and Ram, Ori and Zhang, Dan and Rosen, Evan and others},
  journal={arXiv preprint arXiv:2507.06261},
  year={2025}
}

@article{hendrycks2016gaussian,
  title={Gaussian error linear units (gelus)},
  author={Hendrycks, Dan and Gimpel, Kevin},
  journal={arXiv preprint arXiv:1606.08415},
  year={2016}
}

@article{xiong2024search,
  title={When search engine services meet large language models: visions and challenges},
  author={Xiong, Haoyi and Bian, Jiang and Li, Yuchen and Li, Xuhong and Du, Mengnan and Wang, Shuaiqiang and others},
  journal={IEEE Transactions on Services Computing},
  year={2024},
  publisher={IEEE}
}

@inproceedings{reimers2019sentence,
  title={Sentence-BERT: Sentence Embeddings using Siamese BERT-Networks},
  author={Reimers, Nils and Gurevych, Iryna},
  booktitle={Conference on Empirical Methods in Natural Language Processing and International Joint Conference on Natural Language Processing},
  year={2019},
  organization={Association for Computational Linguistics}
}

@misc{minishlab2024model2vec,
  author = {Tulkens, Stephan and {van Dongen}, Thomas},
  title = {Model2Vec: Fast State-of-the-Art Static Embeddings},
  year = {2024},
  url = {https://github.com/MinishLab/model2vec}
}

@inproceedings{lin2022truthfulqa,
  title={TruthfulQA: Measuring How Models Mimic Human Falsehoods},
  author={Lin, Stephanie and Hilton, Jacob and Evans, Owain},
  booktitle={Annual Meeting of the Association for Computational Linguistics},
  pages={3214--3252},
  year={2022}
}

@inproceedings{rajpurkar2016squad,
  title={SQuAD: 100,000+ Questions for Machine Comprehension of Text},
  author={Rajpurkar, Pranav and Zhang, Jian and Lopyrev, Konstantin and Liang, Percy},
  booktitle={Conference on Empirical Methods in Natural Language Processing},
  pages={2383--2392},
  year={2016}
}

@ONLINE{wikidump,
    author = "Wikimedia Foundation",
    title  = "Wikimedia Downloads",
    url    = "https://dumps.wikimedia.org"
}

@inproceedings{pennington2014glove,
  title={Glove: Global vectors for word representation},
  author={Pennington, Jeffrey and Socher, Richard and Manning, Christopher D},
  booktitle={Conference on Empirical Methods in Natural Language Processing},
  pages={1532--1543},
  year={2014}
}

@article{embeddinggemma,
    title={EmbeddingGemma: Powerful and Lightweight Text Representations},
    author={Schechter Vera, Henrique and Dua, Sahil and Zhang, Biao and Salz, Daniel and Mullins, Ryan and others},
    publisher={Google DeepMind},
    year={2025},
    journal={arxiv preprint arxiv:2509.20354},
}

@article{olmo2024olmo,
  title={2 OLMo 2 Furious},
  author={OLMo, Team},
  journal={arXiv preprint arXiv:2501.00656},
  year={2024}
}

@inproceedings{nam2024using,
  title={Using an llm to help with code understanding},
  author={Nam, Daye and Macvean, Andrew and Hellendoorn, Vincent and Vasilescu, Bogdan and Myers, Brad},
  booktitle={IEEE/ACM International Conference on Software Engineering},
  pages={1--13},
  year={2024}
}

@article{team2024qwen2,
  title={Qwen2 technical report},
  author={Team, Qwen},
  journal={arXiv preprint arXiv:2407.10671},
  volume={2},
  year={2024}
}

@inproceedings{zhu2025collaborative,
  title={Collaborative Retrieval for Large Language Model-based Conversational Recommender Systems},
  author={Zhu, Yaochen and Wan, Chao and Steck, Harald and Liang, Dawen and Feng, Yesu and Kallus, Nathan and Li, Jundong},
  booktitle={The ACM Web Conference},
  pages={3323--3334},
  year={2025}
}

@inproceedings{chen2017zoo,
  title={Zoo: Zeroth order optimization based black-box attacks to deep neural networks without training substitute models},
  author={Chen, Pin-Yu and Zhang, Huan and Sharma, Yash and Yi, Jinfeng and Hsieh, Cho-Jui},
  booktitle={ACM Workshop on Artificial Intelligence and Security},
  pages={15--26},
  year={2017}
}

@article{yao2025black,
  title={Black-Box Guardrail Reverse-engineering Attack},
  author={Yao, Hongwei and Xia, Yun and Shao, Shuo and Shi, Haoran and Qiao, Tong and Wang, Cong},
  journal={arXiv preprint arXiv:2511.04215},
  year={2025}
}

\appendix
\section*{Appendix}
\setcounter{theorem}{0}

\section{Proof of Theorem~\ref{theorem:fisher}}
\label{sec:proof}

Before proving Theorem~\ref{theorem:fisher}, we first introduce and prove a related lemma about the Fisher information of a random variable following a normal distribution, as follows.

\begin{lemma}[Fisher Information for a Normal Distribution]
\label{lemma:normal}
Let $X$ be a random variable following a normal distribution with mean $\mu(\theta)$ and variance $\sigma^2(\theta)$, which are differentiable functions of a parameter $\theta$. The Fisher information of $X$ with respect to $\theta$ is:
\begin{equation}
\mathcal{I}_X(\theta) = \frac{1}{\sigma^2(\theta)} \left( \frac{\partial \mu(\theta)}{\partial \theta} \right)^2 + \frac{1}{2\sigma^4(\theta)} \left( \frac{\partial \sigma^2(\theta)}{\partial \theta} \right)^2.
\end{equation}
\end{lemma}

\begin{proof}[Proof of Lemma~\ref{lemma:normal}]
The probability density function of the normal distribution is:
\begin{equation}
p(x; \theta) = \frac{1}{\sqrt{2\pi\sigma^2(\theta)}} \exp\left(-\frac{(x - \mu(\theta))^2}{2\sigma^2(\theta)}\right).
\end{equation}
The log-likelihood function is:
\begin{equation}
\log p(x; \theta) = -\frac{1}{2}\log(2\pi) - \frac{1}{2}\log(\sigma^2(\theta)) - \frac{(x - \mu(\theta))^2}{2\sigma^2(\theta)}.
\end{equation}
The score function is the partial derivative of the log-likelihood with respect to $\theta$:
\begin{equation}
\frac{\partial}{\partial \theta} \log p(x; \theta) = \frac{(x - \mu(\theta))}{\sigma^2(\theta)}\frac{\partial \mu(\theta)}{\partial \theta} + \frac{(x - \mu(\theta))^2 - \sigma^2(\theta)}{2\sigma^4(\theta)}\frac{\partial \sigma^2(\theta)}{\partial \theta}.
\end{equation}
The Fisher information is the expectation of the square of the score. Squaring the score function and taking the expectation, the cross-product term will have an expectation of $0$. Thus, we only need to consider the expectation of the squared terms. Let $Z=X-\mu(\theta)\sim \mathcal{N}(0, \sigma^2(\theta))$, using the central moments of the normal distribution $\mathbb{E}[Z] = 0$, $\mathbb{E}[(Z)^2] = \sigma^2(\theta)$, and $\mathbb{E}[(Z)^4] = 3\sigma^4(\theta)$, we get:
\begin{equation}
\mathcal{I}_X(\theta) = \mathbb{E}\left[ \left( \frac{(Z)}{\sigma^2(\theta)}\frac{\partial \mu(\theta)}{\partial \theta} \right)^2 \right] + \mathbb{E}\left[ \left( \frac{(Z)^2 - \sigma^2(\theta)}{2\sigma^4(\theta)}\frac{\partial \sigma^2(\theta)}{\partial \theta} \right)^2 \right].
\end{equation}
For the first term, we have:
\begin{equation}
\begin{aligned}
    &\mathbb{E}\left[ \left( \frac{(Z)}{\sigma^2(\theta)}\frac{\partial \mu(\theta)}{\partial \theta} \right)^2 \right] \\ 
    =& \frac{1}{\sigma^4(\theta)} \left( \frac{\partial \mu(\theta)}{\partial \theta} \right)^2 \mathbb{E}[(Z)^2]
    = \frac{1}{\sigma^2(\theta)} \left( \frac{\partial \mu(\theta)}{\partial \theta} \right)^2.
\end{aligned}
\end{equation}
For the second term, we have:
\begin{equation}
    \begin{aligned}
        &\mathbb{E}\left[ \left( \frac{(Z)^2 - \sigma^2(\theta)}{2\sigma^4(\theta)}\frac{\partial \sigma^2(\theta)}{\partial \theta} \right)^2 \right]\\
        =& \frac{1}{4\sigma^8(\theta)} \left( \frac{\partial \sigma^2(\theta)}{\partial \theta} \right)^2 \mathbb{E}[((Z)^2 - \sigma^2(\theta))^2]\\
        =&\frac{1}{4\sigma^8(\theta)} \left( \frac{\partial \sigma^2(\theta)}{\partial \theta} \right)^2 (\mathbb{E}[(Z)^4] - 2\sigma^2(\theta)\mathbb{E}[(Z)^2] + \sigma^4(\theta)) \\
        =& \frac{1}{4\sigma^8(\theta)} \left( \frac{\partial \sigma^2(\theta)}{\partial \theta} \right)^2 (3\sigma^4(\theta) - 2\sigma^4(\theta) + \sigma^4(\theta)) \\
        =& \frac{2\sigma^4(\theta)}{4\sigma^8(\theta)} \left( \frac{\partial \sigma^2(\theta)}{\partial \theta} \right)^2 = \frac{1}{2\sigma^4(\theta)} \left( \frac{\partial \sigma^2(\theta)}{\partial \theta} \right)^2.
    \end{aligned}
\end{equation}
Therefore, we get:
\begin{equation}
    \mathcal{I}_X(\theta) = \frac{1}{\sigma^2(\theta)} \left( \frac{\partial \mu(\theta)}{\partial \theta} \right)^2 + \frac{1}{2\sigma^4(\theta)} \left( \frac{\partial \sigma^2(\theta)}{\partial \theta} \right)^2.
\end{equation}
This completes the proof of Lemma~\ref{lemma:normal}.
\end{proof}

Given Lemma~\ref{lemma:normal}, we then provide the proof of Theorem~\ref{theorem-apd:fisher}.

\begin{theorem}
    \label{theorem-apd:fisher}
    Let $Y = f(WX + K)$, where $f$ is neither a linear nor an affine function, $W$ and $K$ are constants with $W \neq 0$, and $D=\mathrm{d}Y/\mathrm{d}X=Wf'(WX+K)$. Suppose that $X$ follows a zero-mean normal distribution, denoted as $X \sim \mathcal{N}(0, \sigma_X^2)$, and $f''(K)\neq 0$. Under the first-order Taylor approximation, the Fisher information of $Y$ and $D$ with respect to $W$ satisfies:
    \begin{equation}
        \mathcal{I}_D(W) \geq (\frac{c_1^2}{2W^2c_2^2\sigma_X^2}+4) \cdot \mathcal{I}_Y(W).
    \end{equation}
    where $c_1=f'(K), c_2=f''(K)$.
\end{theorem}

\begin{proof}[Proof of Theorem~\ref{theorem-apd:fisher}]
    We first analyze the Fisher information of $Y$ with respect to $W$. Let $Z = WX + K$. Since $X \sim \mathcal{N}(0, \sigma_X^2)$, the random variable $Z$ is also normally distributed with mean $\mu_Z(W) = K$ and variance $\sigma_Z^2(W) = W^2\sigma_X^2$. We can now use Lemma~\ref{lemma:normal} to calculate $I_Z(W)$ with the parameter $\theta = W$. The partial derivatives of the mean and variance of $Z$ with respect to $W$ are:
    \begin{equation}
    \frac{\partial \mu(W)}{\partial W} = 0 \quad \text{and} \quad \frac{\partial \sigma^2(W)}{\partial W} = 2W\sigma_X^2.
    \end{equation}
    Substituting these into the formula from Lemma~\ref{lemma:normal}:
    \begin{equation}
    \begin{aligned}
    &\mathcal{I}_Z(W) = \frac{1}{\sigma^2(W)} \left( \frac{\partial \mu(W)}{\partial W} \right)^2 + \frac{1}{2\sigma^4(W)} \left( \frac{\partial \sigma^2(W)}{\partial W} \right)^2 \\
    =& \frac{1}{W^2\sigma_X^2}*0^2 + \frac{1}{2(W^2\sigma_X^2)^2}(2W\sigma_X^2)^2 = \frac{4W^2\sigma_X^4}{2W^4\sigma_X^4} = \frac{2}{W^2}.
    \end{aligned}
    \end{equation}
    Since $Y=f(Z)$ and $f$ is independent of the parameter $W$, according to the Data Processing Inequality (DPI), the amount of information in data does not increase after any processing that is independent of the parameters. Therefore, we have:
    \begin{equation}
        \mathcal{I}_Y(W)\leq\mathcal{I}_Z(W)=\frac{2}{W^2}.
    \end{equation}
    Then, we analyze the Fisher information of $D=Wf'(WX+K)$ with respect to $W$. Expand $f'$ about $Z=K$:
    \begin{equation}
    f'(Z) = c_1 + c_2(Z-K) + R_2(Z),
    \end{equation}
where $c_1:=f'(K), c_2=f''(K)$, and $R_2(Z)$ is the Taylor remainder. Neglecting the remainder term $R_2$ (this is the first-order approximation), and using $Z-K = W X$, we obtain the approximate linear form as follows:
\begin{equation}
\label{eq:D-approx}
D \approx D_{\mathrm{approx}} := c_1 W + c_2 W^2 X.
\end{equation}
Under this approximation, $D_{\mathrm{approx}}$ follows a normal distribution $\mathcal{N}(\mu_D(W), \sigma^2_D(W))$ because $X$ also follows a normal distribution.
From Eq. \eqref{eq:D-approx}, we can calculate the mean and variance of $D_{\mathrm{approx}}$ as functions of $W$:
\begin{equation}
\begin{aligned}
&\mu_D(W) := \mathbb{E}[D_{\mathrm{approx}}] = W c_1, \\
&\sigma^2_D(W) := \operatorname{Var}(D_{\mathrm{approx}}) = W^4 c_2^2 \sigma_X^2,
\end{aligned}
\end{equation}
since $\mathbb{E}[X]=0$ and $\operatorname{Var}(X)=\sigma_X^2$. Differentiate the mean and variance w.r.t. $W$:
\begin{equation}
\begin{aligned}
    \frac{\partial \mu_D(W)}{\partial W} &= c_1, \ \ \frac{\partial \sigma^2_D(W)}{\partial W}= 4 W^3 c_2^2 \sigma_X^2.
\end{aligned}
\end{equation}
Invoking Lemma~\ref{lemma:normal} with the parameter $\theta=W$, mean $\mu_D(W)$ and variance $\sigma^2_D(W)$, we obtain the Fisher information of the approximate normal model:
\begin{equation}
\begin{aligned}
\mathcal{I}_D(W)
&\approx
\frac{1}{\sigma^2_D(W)}\Big(\frac{\partial \mu_D(W)}{\partial W}\Big)^2
+
\frac{1}{2\sigma^4_D(W)}\Big(\frac{\partial \sigma^2_D(W)}{\partial W}\Big)^2 \\[4pt]
&=
\frac{c_1^2}{W^4c_2^2\sigma_X^2}
+
\frac{1}{2\,W^8c_2^4\sigma_X^4}\,\Big(4W^3c_2^2\sigma_X^2\Big)^2 \\
&=
\frac{c_1^2}{W^4c_2^2\sigma_X^2}+\frac{8}{W^2} =
\frac{2}{W^2}(\frac{c_1^2}{2W^2c_2^2\sigma_X^2}+4).
\end{aligned}
\end{equation}
Therefore, we have the following conclusion: in this approximation,
\begin{equation}
    \mathcal{I}_D(W) \geq (\frac{c_1^2}{2W^2c_2^2\sigma_X^2}+4) \cdot \mathcal{I}_Y(W).
\end{equation}
This completes the proof of Theorem~\ref{theorem-apd:fisher}.
\end{proof}

\begin{figure}[t]
    \centering
    \chatbox[Example Queries in Query Set $\mathcal{Q}$]{
    \textbf{Base Query:} \\
    Complete the following code: def will\_it\_fly(q,w): ''' Write a function that returns True if the object q will fly, and False otherwise. The object ... \\
    \textbf{Perturbed Query:} \\
    1. Complete the following code : def will\_it\_fly ( q , w ) : `` ' Write a function that returns \red{Fact} if the object q will \red{flies} , and False \red{not} . The object \\
    2. Complete the following code : def will\_it\_fly ( q , w ) : `` ' \red{Book} a function that returns \red{Fact} if the object q will fly , and False \red{simply} . The object
    }
    \vspace{-0.5em}
    \caption{Example queries in the Query Set $\mathcal{Q}$.}
    \label{fig:example}
    \vspace{-1em}
\end{figure}

\begin{figure}[t!]
\vspace{-8pt}

\begin{algorithm}[H]
\caption{\name Query Set Construction}
\label{alg:query_construction}
\footnotesize

\begin{algorithmic}[1]
\renewcommand{\arraystretch}{1.2}

    \REQUIRE $D_{source}$: Source dataset for base queries (\eg, HumanEval) \\
             \hspace*{1.2em} $n$: Number of base queries to select \\
             \hspace*{1.2em} $m$: Number of perturbed queries to generate per base query \\
             \hspace*{1.2em} $r$: Number of words to replace in each perturbation \\
             \hspace*{1.2em} $k$: Number of top semantically similar candidates to consider \\
             \hspace*{1.2em} $E_{word}$: Pre-trained word embedding model (\eg, GloVe)

    \ENSURE $Q$: The query set containing base and perturbed queries
    
    \STATE $Q_b \gets \emptyset$ \textit{\# Initialize set for base queries}
    \STATE $Q \gets \emptyset$ \textit{\# Initialize the final query set}
    
    \FOR{$i \gets 1$ \textbf{to} $n$} 
        \STATE $x_i \gets \text{SampleAndFormat}(D_{source})$ \textit{\# e.g., ``Complete the code: [snippet]''}
        \STATE $Q_b \gets Q_b \cup \{x_i\}$
    \ENDFOR
    \STATE $Q \gets Q \cup Q_b$ \textit{\# Add all base queries to the final set}
    
    \FOR{each base query $x_i \in Q_b$}
        \FOR{$j \gets 1$ \textbf{to} $m$}
            \STATE $x'_{i,j} \gets x_i$ \textit{\# Create a copy of the base query}
            \STATE $W_{replaceable} \gets \text{SelectRandomWords}(x_i, r)$ \textit{\# Select r words to replace}
            
            \FOR{each word $w \in W_{replaceable}$}
                \STATE $C_{candidates} \gets E_{word}.\text{TopKSimilar}(w, k)$ \textit{\# Find top-k similar words}
                \STATE $w_{replacement} \gets \text{RandomChoice}(C_{candidates})$
                \STATE $x'_{i,j} \gets \text{ReplaceWord}(x'_{i,j}, w, w_{replacement})$
            \ENDFOR
            
            \STATE $Q \gets Q \cup \{x'_{i,j}\}$ \textit{\# Add the perturbed query to the final set}
        \ENDFOR
    \ENDFOR
    
    \RETURN $Q$
\end{algorithmic}
\end{algorithm}
\vspace{-2.5em}
\end{figure}

\begin{figure}[t!]
\vspace{-8pt}
\begin{algorithm}[H]
\renewcommand{\arraystretch}{1.2}
\caption{\name Fingerprint Generation}
\label{alg:fingerprint_generation}
\footnotesize


\begin{algorithmic}[1]
\setlength{\itemsep}{1pt}
    
    \REQUIRE $\mathcal{M}$: The target black-box LLM \\
             \hspace*{1.2em} $Q_b = \{x_i\}_{i=1}^n$: Set of base queries \\
             \hspace*{1.2em} $Q_p = \{\{x'_{i,j}\}_{j=1}^m\}_{i=1}^n$: Set of perturbed queries \\
             \hspace*{1.2em} $E$: Pre-trained sentence embedding model \\
             \hspace*{1.2em} $t$: Number of times to query the model for each input \\
             \hspace*{1.2em} $\alpha$: Regularization hyperparameter for Ridge Regression
    
    \ENSURE $J_\mathcal{M}$: The final fingerprint (aggregated Jacobian matrix)
    
    \STATE $J_{list} \gets []$ \textit{\# Initialize list to store local Jacobians}
    
    \FOR{each base query $x_i \in Q_b$}
        \STATE $\Delta X_i, \Delta Y_i \gets [], []$ \textit{\# Initialize lists for input/output differences}
        
        \STATE \textit{\# Get stable embedding for the base query's output}
        \STATE $\{y_{i,k}\}_{k=1}^t \gets \text{QueryModelMultipleTimes}(\mathcal{M}, x_i, t)$
        \STATE $\bar{e}_{y_i} \gets \frac{1}{t} \sum_{k=1}^{t} E(y_{i,k})$
        \STATE $e_{x_i} \gets E(x_i)$
        
        \FOR{each perturbed query $x'_{i,j}$ for $x_i$}
            \STATE \textit{\# Get stable embedding for the perturbed query's output}
            \STATE $\{y'_{i,j,k}\}_{k=1}^t \gets \text{QueryModelMultipleTimes}(\mathcal{M}, x'_{i,j}, t)$
            \STATE $\bar{e}_{y'_{i,j}} \gets \frac{1}{t} \sum_{k=1}^{t} E(y'_{i,j,k})$
            \STATE $e_{x'_{i,j}} \gets E(x'_{i,j})$
            
            \STATE \textit{\# Calculate input and output embedding differences}
            \STATE $\Delta x_{i,j} \gets e_{x'_{i,j}} - e_{x_i}$
            \STATE $\Delta y_{i,j} \gets \bar{e}_{y'_{i,j}} - \bar{e}_{y_i}$
            \STATE Append $\Delta x_{i,j}$ to $\Delta X_i$
            \STATE Append $\Delta y_{i,j}$ to $\Delta Y_i$
        \ENDFOR
        
        \STATE \textit{\# Estimate the local Jacobian using Ridge Regression}
        \STATE $J_i \gets \arg\min_{J} \sum_{j=1}^{m} ||\Delta y_{i,j} - J \Delta x_{i,j}||_2^2 + \alpha ||J||_F^2$
        \STATE Append $J_i$ to $J_{list}$
    \ENDFOR
    
    \STATE $J_\mathcal{M} \gets \frac{1}{n} \sum_{i=1}^{n} J_{list}[i]$
    \RETURN $J_\mathcal{M}$
\end{algorithmic}
\end{algorithm}
\vspace{-2.5em}
\end{figure}

\section{Implementation Details}
\label{apd:details}

\firstpartitle{Introduction to \bench} \bench is a benchmark that facilitates a unified evaluation of LLM fingerprinting methods~\cite{shao2025sok}. It includes 7 base model lineages, including Qwen2.5-7B, Qwen2.5-14B, Llama3.1-8B, Mistral-7B-v0.3, Gemma-2-2B, TinyLlama-1.1B-v1.0, and Llama-2-7B. \bench also implements 13 different post-development techniques that may be adopted to adjust the pre-trained base models. These techniques can significantly affect the behavior of the LLMs. These techniques include: \textbf{(1)} pretraining, \textbf{(2)} instructional tuning, \textbf{(3)} fine-tuning, \textbf{(4)} parameter-efficient fine-tuning, \textbf{(5)} quantization, \textbf{(6)} model merging, \textbf{(7)} distillation, \textbf{(8)} general-purpose system prompts, \textbf{(9)} role-playing system prompts, \textbf{(10)} chain-of-thought prompts, \textbf{(11)} sampling strategies, \textbf{(12)} retrieval-augmented generation (RAG), \textbf{(13)} adversarial manipulation. \bench includes a total number of 149 model instances.

\partitle{Detailed Settings for \name} To generate the base queries, we sample code snippets from HumanEval, take the first 20 words, and then prepend them with ``Complete the following code:'' to form a query for a completion task. When conducting word substitutions, we use GloVe~\cite{pennington2014glove} as the word embedding model. For each base query, we randomly replace 3 words with one of the top-10 similar words. The example queries are shown in Figure~\ref{fig:example} and the pseudocodes of query set construction and fingerprint generation in \name are in Algorithm~\ref{alg:query_construction}\&\ref{alg:fingerprint_generation}. The coefficient $\alpha$ in Ridge Regression is set to $0.001$. In our experiments, we utilize NVIDIA RTX 4090 GPUs with a total of 96GB VRAM.

\partitle{Detailed Settings for the Baseline Fingerprinting Methods} In our experiments, we take 1 white-box method and 4 state-of-the-art black-box methods as our baseline. For fair comparison, following \cite{zhang2025reef}, we also set the query limit to $q=200$ while applying these methods. The detailed settings are as follows.
\begin{itemize}
    \item \textbf{REEF}~\cite{zhang2025reef}: REEF is a forward-pass white-box fingerprinting method that relies on the intermediate representations as fingerprints. We strictly follow the original setting and sample 200 queries from TruthfulQA~\cite{lin2022truthfulqa} as the input queries.
    \item \textbf{LLMmap}~\cite{pasquini2025llmmap}: LLMmap is an untargeted black-box fingerprinting method. We adopt the 8 carefully-crafted prompts as the input queries and utilize the pre-trained open-set feature extraction model provided in its open-source repository.
    \item \textbf{MET}~\cite{gao2025model}: MET is also a untargeted black-box fingerprinting method that uses text completion prompts as input queries. The original MET uses 25 samples. Therefore, we run MET 8 times, leading to a total of 200 queries.
    \item \textbf{SEF}~\cite{shao2025sok}: SEF is an untargeted black-box fingerprinting method that straightforwardly uses the sentence embeddings of the outputs. Following original setting, we use Qwen3-Embedding-4B as the sentence embedding model and sample 200 queries from four different datasets, DROP~\cite{dua2019drop}, Ethics~\cite{hendrycks2021ethics}, PubMedQA~\cite{jin2019pubmedqa}, and HumanEval~\cite{chen2021humaneval}.
    \item \textbf{TRAP}~\cite{gubri2024trap}: TRAP is a targeted black-box fingerprinting method that leverages GCG~\cite{zou2023universal} to generate targeted query-response pairs. For each query, we use GCG to optimize 100 iterations.
\end{itemize}

\partitle{Detailed Settings of Ablation Studies} For the ablation study of different query data, except the HumanEval dataset~\cite{chen2021humaneval}, we select four other datasets as the source of the base queries, including the generative QA dataset TruthfulQA~\cite{lin2022truthfulqa}, the extractive QA dataset SQuAD~\cite{rajpurkar2016squad}, the complicated reasoning QA dataset DROP~\cite{dua2019drop}, and a text completion dataset Wikipedia~\cite{wikidump}. For the Wikipedia dataset, similar to HumanEval, we sample a text snippet with 20 words and prepend a completion instruction (\ie, ``Continue the following text:'') to construct a completion prompt. For the ablation study of different sentence embedding models, we select the following 5 models: all-mpnet-base-v2 (MPNet), all-MiniLM-L6-V2 (MiniLM)~\cite{reimers2019sentence}, potion-multilingual-128M (Model2Vec)~\cite{minishlab2024model2vec}, Qwen3-Embedding-0.6B (Qwen3)~\cite{qwen3embedding}, and EmbeddingGemma-300m (EmbeddingGemma)~\cite{embeddinggemma}. They are all classic or state-of-the-art sentence embedding models with different embedding dimensions.

\partitle{Detailed Settings of Adaptive Attacks} For input paraphrasing attacks, we utilize three small LLMs: Qwen2-0.5B-Instruct~\cite{team2024qwen2}, OLMo-2-0425-1B-Instruct~\cite{olmo2024olmo}, and Gemma-2-2b-it~\cite{gemma2}, to rewrite the input queries. The experiments of adaptive attacks are conducted on Qwen2.5-7B-Instruct and Llama3.1-8B-Instruct models.

\partitle{Detailed Settings of the Overhead Evaluation} The overhead evaluation in Section~\ref{sec:overhead} is conducted on four derivative models from the Qwen2.5-7B lineage. For fair comparison, we utilize 1 NVIDIA RTX 4090 48GB GPU to compute the runtime of all the evaluated LLM fingerprinting methods.

\begin{table}[t]
\centering
\tabcolsep=2.8mm
\renewcommand{\arraystretch}{1.1}
\caption{The fingerprint similarities across different LLMs while auditing closed-source LLMs using \name.}
\label{tab:closed_source}
\vspace{-0.5em}
\scalebox{0.85}{
\begin{tabular}{l c c c c}
\toprule
\textbf{Type} & \textbf{Deriative} & \multicolumn{3}{c}{\textbf{Independent Models}} \\
\cmidrule(lr){2-2} \cmidrule(lr){3-5}
\textbf{Models} & GPT-4o & Gemini-2.5 & Claude-4.5 & DeepSeek-V3   \\
\midrule
GPT-4o & 0.943 & 0.529 & 0.570 & 0.578 \\
\bottomrule
\end{tabular}
}
\vspace{-1em}
\end{table}

\section{Case Study of Closed-source Audit}

In this section, we include a case study of auditing closed-source LLMs using \name. Specifically, we utilize GPT-4o-mini (GPT-4o) as the source model and the same GPT-4o-mini API as the derivative, and take three other LLMs, Gemini-2.5-Flash-Lite (Gemini-2.5), Claude-Haiku-4.5 (Claude-4.5), and DeepSeek-V3, as independent models. The results in Table~\ref{tab:closed_source} show that the fingerprint similarity of fingerprints from the same GPT-4o-mini model is $0.943$, which is close to $1$, while other independent LLMs exhibit low similarity ($\approx 0.50$). These results demonstrate the feasibility and potential of our \name to audit closed-source, black-box LLMs.

\vfill
\section{Ethic Statement}

The unauthorized usage of LLM has raised significant concerns about copyright infringements, while LLM fingerprinting offers a promising solution. In this paper, we propose a novel black-box fingerprinting method, \name. Our method is designed for the purely defensive purpose of copyright auditing and does not introduce new threats or vulnerabilities. Furthermore, our work utilizes publicly available datasets, such as HumanEval, for query construction and does not infringe on the privacy of any individual. The research does not involve any human subjects. As such, this work does not raise any major ethical issues.

\end{document}